\documentclass[10pt,journal,compsoc]{IEEEtran}
\ifCLASSOPTIONcompsoc
  \usepackage[nocompress]{cite}
\else
  \usepackage{cite}
\fi
\ifCLASSINFOpdf
   \usepackage{graphicx}
\else
\fi

\usepackage{etex}        % required for bibliography
\usepackage{amsfonts}
\usepackage{times}
\usepackage{latexsym,amsfonts,amsbsy}
\usepackage{amsmath}
\usepackage{amssymb}
\usepackage{hyperref}
\usepackage{textcomp}
\usepackage{gensymb}

\usepackage[absolute,overlay]{textpos}
\usepackage{amsbsy}
\usepackage{stmaryrd}
\usepackage{MnSymbol}

\usepackage{comment}
\usepackage{mathtools}
\usepackage{amsthm}
\DeclareMathOperator*{\argmin}{arg\,min}

\usepackage[ruled,linesnumbered]{algorithm2e}
\usepackage{tabularx,booktabs}
\usepackage{epstopdf}
\allowdisplaybreaks
\usepackage{mwe}    % loads »blindtext« and »graphicx«
\usepackage{color}
\usepackage{cite}
\usepackage{float}
\usepackage{tikz}
\usetikzlibrary{arrows,automata}
\usetikzlibrary{calc}
\usetikzlibrary{shapes.geometric}
\usetikzlibrary{positioning}

\tikzstyle{startstop} = [rectangle, rounded corners, minimum width=3cm, minimum height=1cm,text centered, draw=black]
\tikzstyle{process} = [rectangle, rounded corners,minimum width=3cm, minimum height=.8cm,, text centered, draw=black]
\tikzstyle{decision} = [diamond, minimum width=2cm, minimum height=1cm, text centered, draw=black]
\tikzstyle{arrow} = [thick,->,>=stealth]
\tikzstyle{blank} = [node distance=1cm]
\tikzstyle{block} = [rectangle, draw, fill=blue!20, text centered, rounded corners, minimum height=1em]
\tikzstyle{line} = [draw, -latex']
\tikzstyle{cloud} = [draw, ellipse,fill=red!20, node distance=4cm,minimum height=1em]

\begin{document}
\title{Ensuring Privacy in Location-Based Services: A Model-based Approach}

\author{Alireza~Partovi, Wei~Zheng~\IEEEmembership{Student~Member,~IEEE,}
        Taeho~Jung,~\IEEEmembership{Member,~IEEE,}
        and~Hai~Lin,~\IEEEmembership{Senior~Members,~IEEE}% <-this % stops a space
\IEEEcompsocitemizethanks{
\IEEEcompsocthanksitem A. Partovi, Wei~Zheng, and H. Lin are with the Department of Electrical Engineering, University of Notre Dame, Notre Dame,	IN, 46556 USA.\protect\\
E-mails: $\{$apartovi@nd.edu,wzheng1,hlin1$\}$@nd.edu.
\IEEEcompsocthanksitem T. Jung  is with the Department of Computer Science and Engineering, University of Notre Dame, Notre Dame, IN, 46556 USA. \protect\\
E-mail: tjung@nd.edu.
}% <-this % stops an unwanted space
}


\markboth{Preprint. Under Review.}{}
%{Shell \MakeLowercase{\textit{et al.}}: Bare Demo of IEEEtran.cls for Computer Society Journals}



\newtheorem{definition}{Definition}
\newtheorem{problem}{Problem}
\newtheorem{rem}{Remark}
\newtheorem{alg}{Algorithm}
\newtheorem{thm}{Theorem}
\newtheorem{cor}{Corollary}
\newtheorem{prop}{Proposition}
\newtheorem{lemma}{Lemma}
\newtheorem{example}{Example}

\IEEEtitleabstractindextext{%
\begin{abstract} %for TMC
    In recent years, the widespread of mobile devices equipped with GPS and communication chips has led to the growing use of location-based services (LBS) in which a user receives a service based on his current location. The disclosure of user's location, however, can raise serious concerns about user privacy in general, and location privacy in particular which led to the development of various location privacy-preserving mechanisms aiming to enhance the location privacy while using LBS applications. In this paper, we propose to model the user mobility pattern and utility of the LBS as a Markov decision process (MDP), and inspired by probabilistic current state opacity notation, we introduce a new location privacy metric, namely $\epsilon-$privacy, that quantifies the adversary belief over the user's current location.  We exploit this dynamic model to design a LPPM that while it ensures the utility of service is being fully utilized, independent of the adversary prior knowledge about the user, it can guarantee a user-specified privacy level can be achieved for an infinite time horizon. 
The overall privacy-preserving framework, including the construction of the user mobility model as a MDP, and design of the proposed LPPM, are demonstrated and validated with real-world experimental data.
\end{abstract}

}

\maketitle

\IEEEdisplaynontitleabstractindextext

% \fi
\IEEEpeerreviewmaketitle

\IEEEraisesectionheading{\section{Introduction}\label{sec:introduction}}

\IEEEPARstart{A}{s}  a result of recent technological advances  in sensing and tracking, and the widespread of mobile devices with significant computational and communication capabilities, the location-based applications become increasingly popular.
%A variety of these applications have been successfully developed in which a user receives a service based on his current location. 
Example of LBS applications in smartphones are  mobile navigation, ride-sharing,  location-aware social networks, and location-based contextual advertising,  and dining recommendation.

Even though LBS are providing great benefits to the users, the exposure of the user's location raises major personal privacy concerns. The LBS servers or third party systems receive and  store user location data could use them to infer user's precise location and track his point of interest, giving rise to a variety of malicious activities \cite{chatzikokolakis2017methods}. Examples of these threats are \textit{tracking threats}: an adversary identifies user's mobility pattern and predict his future locations \cite{jan2000using}, \textit{identification threats}: an adversary uses user's locations to infer his identity from an anatomized database \cite{de2013unique}, and \textit{profiling threats}:   an adversary uses user's location of interest for profiling him in variety of sensitive information such as political view and health condition \cite{ashbrook2003using}. 

 Large body research has studied these  privacy issues and various privacy protection methods have developed to allow users to utilize LBS while limiting the leakage of users' confidential information. These methods are known as  location privacy protection mechanisms (LPPMs) that can be roughly divided into two main classes: \textit{identity anonymization} techniques   and \textit{location perturbation} techniques \cite{chatzikokolakis2017methods}. 
 %The basic idea in LPPMs  is to depersonalized data by degrading  linkability of users’ identities from their true location information. 
Anonymization techniques protect user privacy by dissociating the user's real identity from  his location-based information. 
This usually is done by a third party anonymizer  that replaces the identity of users with  temporary identifiers, namely  pseudonyms \cite{beresford2004mix}.
%\cite{pfitzmann2001anonymity,beresford2003location,beresford2004mix,palanisamy2011mobimix}. 
However, it turns out that merely removing or replacing user identity does not provide a strong privacy since spatio-temporal characteristics of the data can still help an adversary to track and re-identify the anonymous users \cite{hoh2006enhancing}. Thus, in addition to concealing users' identity, it is also important to obfuscate the users' position.

Location obfuscation mechanisms protect users privacy by deliberately degrading  precision of  users' location information in a way that the service can still be carried out to some acceptable extent without revealing users' true location \cite{riaz2018location}. 
Generally, this
is achieved by \textit{spatial obfuscation} techniques such as adding some noise \cite{andres2012geo,elsalamouny2016differential},  using dummy locations  \cite{niu2013pseudo,lu2008pad} or by  \textit{spatial cloaking} that  essentially enlarges the user's queried region \cite{gruteser2003anonymous}.  
%These methods however are bound to fail in presence 
%Despite extensive previous studies on location privacy and LPPM mechanisms, 

These methods offer location  privacy by  increasing the adversary uncertainty about the user's current position.  
 However, often a strong adversary has prior knowledge about the user's movement pattern and can strategically update her belief based on the user's service queries and eventually can reduce her uncertainty over the user true locations %\cite{shokri2011quantifying}.  
 \cite{shokri2011quantifying,chatzikokolakis2017methods}.  
 Therefore,  in addition to concealing and obfuscating the user's current position, the LPPM  should take into account the adversary inference capability for the current  and   future observations.
 Furthermore, the LPPMs proposed in these methods typically assume the adversary prior knowledge does not violate the user desired privacy level which may not be a valid assumption if the adversary has accurate background knowledge about the user's behavior. For instance, a database derived from social media applications shows that the user's check-in locations  combined with other publicly available information such as the popularity of these locations, can effectively be utilized by an adversary to build an strong background knowledge about the user mobility pattern \cite{ahuja2019utility}.

To address these concerns, in this paper, we aim to provide a model-based location privacy for a  user who makes continuous queries from the LBS server.  Without loss of generality, we assume that the adversary is the LBS server who has knowledge of the user mobility pattern but is not capable of observing the user's real-time true locations. We design a LPPM that offers all-time privacy protection without having access to the adversary prior knowledge, and furthermore, we show that even if the prior knowledge does not meet the user-specified privacy level, the LPPM can still deceive the adversary to eventually respect the user's privacy standard.
%against an adversary who has the intention of inferring the user real-time locations.
%and is capable to strategically update his estimation from  the location information received in the user's queries.
Motivated by the event-driven nature of the mobile user's mobility model in LBS,  we propose to use a model-based location privacy-preserving framework. In particular,  we construct a Markov decision process that represents the user mobility patterns and the LBS utility model. In order to characterize the user location privacy, we adapt the notion of probabilistic current state opacity (CSO) that has been studied for the MDPs  and discrete event systems  \cite{Wu2020}
and introduce a related new notion called $\epsilon-$\textit{privacy} that captures the adversary uncertainty on the user's current location. In this setup, the user locations are the MDP states, and therefore location privacy is protected if and only if the constructed
MDP meets the $\epsilon-$privacy criteria.

The rest of this paper is organized as follows. In the subsequent section, we discuss the related works. 
 %we present some preliminary definitions related to  MDP in Section \ref{sec:preliminary}. 
Section \ref{sec:system_model} describes the process of constructing a MDP representing the user mobility patterns and the LBS utility model. Section \ref{sec:privacy_model} introduces the adversary thread model, and $\epsilon-$privacy metric, and furthermore studies the limitation of other popular location privacy notations including \textit{entropy}, \textit{expected inference error}, and \textit{differential privacy} for the localization attacks. Section \ref{sec:LPP_problem} presents the location privacy-preserving problem based on the proposed $\epsilon-$privacy metric.
Sections  \ref{sec:all_time_LPPM} provides the LPPM design process and addresses its computational complexity. 
Section \ref{sec:case_study} demonstrates  applicability of the proposed LPPM on an experimental dataset. The paper is concluded in 
Section \ref{sec:conclusion}.    

\section{Related Works} \label{sec:related_work}
Location privacy has been an active field of research over the past decade and various LPPMs have been proposed to protect user privacy.
An early approach for preserving privacy is to replace the identity of a user with a pseudonym \cite{pfitzmann2001anonymity}.  Some other methods propose to frequently change the pseudonyms when the users are within areas called mix-zones \cite{beresford2004mix}.
%\cite{beresford2003location,beresford2004mix,palanisamy2011mobimix}.
However, these approaches may fail in the context of LBS, since an adversary can de-anonymized and re-identify the anonymous users by correlating his reported location's information and the user background information \cite{hoh2006enhancing}. 
As a consequence,  in addition to anonymizing users' identities, the users' locations also should be perturbed before being supplied to the LBS server.

A common technique to perturb users' location is to reduce the precision of the location information in the spatial and temporal domains. This can often be achieved by cloaking techniques that essentially reduce the granularity of the users'  information. % or by deliberately adding noise in the data \cite{ardagna2007location}.
The $k-$anonymity privacy protection model is introduced based on the  cloaking technique \cite{sweeney2002k,gruteser2003anonymous}. This method makes the users' identity indistinguishable within a group of $k-1$ other users that are in the same spatial cloaking area. 
The $k-$anonymity privacy condition, however, only helps on user  \textit{query anonymity}, that is a private property to protect the association between users and queries, while it does not prevent the disclosure of the link between the user and his spatio-temporal data, namely \textit{location privacy}. Additionally, the $k-$anonymity privacy notation was shown to be vulnerable to the presence of an adversary with certain prior knowledge about users visiting locations \cite{shokri2010unraveling}. 

Sometimes being indistinguishable from the other members of a group is not sufficient to guarantee users privacy. For instance, when the entire group has the same privacy concern for being in a sensitive location that may leak their confidential information
%, being indistinguishable among other members of the group does not offer privacy.
To cope with this problem, some papers propose to provide location privacy by making the user's sensitive location indistinguishable from other landmarks. \cite{gruteser2004protecting} introduces a $k-$area cloaking mechanism that ensures a user's sensitive location is concealed by a region that covers at least $k-1$ other sensitive areas. \cite{bamba2008supporting,xue2009location}
%\cite{bamba2008supporting}
proposes principle of  location $l-$diversity. This mechanism provides location privacy by ensuring the user query can be linked to at least $l$ semantically different location objects, such that each of these has a probability $1\slash l$ to be the true one.  These spatial cloaking methods   simply obfuscates the user sensitive location into an uncertainty region, and therefore
 are bound to fail in the presence of a strong adversary with inference capability. Such an adversary with prior knowledge about the user's location can utilize the user's queries to identify the user's true location  \cite{chatzikokolakis2017methods}. 
In contrast, in this paper, we aim to provide location privacy independent of the adversary prior knowledge.

Differential privacy-based approaches protect user privacy independent of the adversary prior knowledge by adding controlled noise to the query outcome. 
\cite{dewri2012local} proposed to combine  differential privacy with $k-$anonymity.
Following this work,  \cite{andres2012geo,andres2013geo} 
%\cite{andres2012geo} 
generalized the differential privacy with arbitrary metric and  developed a Planar Laplace mechanism to achieve $\epsilon-$geo-indistinguishability. 
%\cite{elsalamouny2016differential,andres2012geo,andres2013geo}. 
In these approaches, the LPPM ensures  location privacy by restricting the adversary knowledge gain about the  true location. 
%However, it does not capture what information an adversary is obtained about the user’s true locations.
However, in presence of continues queries, when a user releases the perturbed locations, he may not know how close the adversary’s estimate will get to his secret locations, despite the differential privacy ensures that the relative gain of knowledge for the adversary is bounded.

LPPMs based on distortion privacy  address this issue \cite{shokri2011quantifying}. These methods characterize  user privacy based on the error of inferring the user’s secret location from the reported location information. This, however, requires an assumption of knowing the adversary prior information and is not robust to the adversaries with arbitrary prior knowledge \cite{shokri2015privacy}.

Recently a model-based LPPM is introduced in \cite{WU201433}. This work presents location privacy in the context of a formal method, where the user mobility pattern is modeled by a deterministic discrete event system with the state representing the user's locations.
In this scheme, location privacy is characterized by current-state opacity notation, and opacity enforcement technique is used to guarantee the user location privacy. The proposed LPPM, however, relies  on an unrealistic security assumption that the  user  and the adversary are behaving deterministically \cite{mathew2012predicting}.

To address these issues, we propose to use a model-based LPPM that
characterizes the location privacy based on a  Bayesian adversary who has access to the user mobility model. Through this setup, the LPPM can model the adversary inference dynamic and track the adversary knowledge over the user secrete locations.
Furthermore, we drive the necessary and sufficient conditions over the adversary inference dynamic that guarantees the user-defined privacy level can be achieved.  By incorporating this information in the LPPM design, we develop a privacy-preserving mechanism that randomizes the obfuscated reported locations such that the adversary estimation of the user's secret locations never violates the user-defined privacy requirement.
 %furthermore even if b0 is unsafe ....
\section{User Mobility Model} \label{sec:system_model}
We consider a  mobile user who is required to share his location to receive some information from the service providers. Due to the user privacy concerns, the true locations of the user are required to be kept private to the user, and noisy locations are released to the service provider, which can be visible to an adversary. The adversary, therefore, is assumed to have knowledge of all historically released locations from the users, and hence can associate each mobile user with a mathematical model representing his mobility patterns. Examples are \cite{wang2019next} and  \cite{montazeri2016achieving} that considered Markov chain, and \cite{xiao2017loclok} that uses a hidden Markov chain to represent the user mobility patterns.

In our LPPM framework, the users' location correlations and the quality of service are modeled by a MDP which is assumed to be public and hence accessible by adversaries. In the following, we will formally define MDP, and show how a user mobility pattern can be represented by a MDP. In Section \ref{sec:case_study}, we will further illustrate this procedure based on a real-world dataset.

\begin{definition}
 A Markov decision process (MDP) is a tuple  $M=(S,A,T,p_0,u)$ where $S=\{s_1,\dots,s_n\}$ is a  set of states, and $A=\{a_1,\dots,a_m\}$ is a  set of actions.  $T(s,a,s'):\text{Prob}\left[s' \mid s,a\right]$,  is the probability of transition from state $s$ to $s'$ with action $a \in A$.    The initial state distribution is
$p_0:S \to \left[0,1 \right]$. Given $s \in S$, and  $a\in A$, the utility function is $u(s,a) \in \mathbb{R}$. We denote the set of available actions at state $s$ by $A(s)$.
%The instantaneous utility  is $\mathbf{u} \in \mathbb{R}^{n,m}$ that contains the elements $u(s,a)$. 
%If the reward function $R_t(s,a)$ is time independent, we simply use $u(s,a)$ and $\mathbf{R}$ when the matrix form is intended. 
\end{definition}
We assume the MDP has a finite state space, finite action space, and bounded reward function.
Throughout the paper, we often use $T^a_{ss'}$ as shorthand of $T(s,a,s')$, $\sum_s$ as shorthand for $\sum_{s \in S}$,  $|S|$ as the cardinality  of set $S$, and $A - B$ as a set difference.  Let's denote probability distribution over set $S$ by $\Delta(S)=\{p \in \mathbb{R}^{n} \mid \sum_s p(s)=1, p(s) \ge 0, \forall s \in S\}$,  and we call  \textit{probability vector} $p \in \Delta(S)$ has a uniform distribution, if $p(s)= {1}\slash{|S|} $ for all $s \in S$.  

%\text{Prob} for the probability of random variable,
\subsection{User Point of Interests}
LBS often store users' mobility traces that contain temporal and spatial data of users' visiting places. 
These location traces can be used for statistical analyses to obtain users' typical mobility patterns and in particular to extract the users'  points of interest (POIs),
that the adversary can use to infer a variety of confidential information about the user \cite{ashbrook2003using}. 
 In this paper, we consider a Markov decision process  mobility model, where its states are the users' POIs. Let's denote a user POIs by a finite set  $S$, where each $s \in S$ represents a predetermined spatial data of POIs.

\begin{rem}
 To simplify the presentation, we define the state set $S$ only based on the user's spatial information, however, it would be easy to incorporate other information of user mobility patterns such as timestamps. For instance, we can include timestamps $t_k$, in the state set as $\hat{S}=\{ (s,t_k) \mid s\in S, k \in N \}$, where $N$ represents a finite-time horizon.
\end{rem}

\subsection{Location Release Mechanism} \label{subsec:Loc_Release_Mech}
%To protect  user's privacy, it is important to obfuscate user's true location before releasing it to a LBS server which potentially can be an adversary.  
Generally location obfuscation is achieved by {spatial obfuscation} techniques such as   using dummy locations  \cite{niu2013pseudo,lu2008pad}, {spatial cloaking} \cite{gruteser2003anonymous,chow2011spatial}, or by  adding noise  \cite{andres2012geo,elsalamouny2016differential}. In our privacy protection framework, the location obfuscation mechanism for each user's POI can be seen as an action of reporting the obfuscated position of that location. We therefore consider  the MDP action set $A$ as a set of user's obfuscated spatial information  that are reported to the LBS server and hence are observable to the adversary.

\begin{rem}
 One of our objective here is to demonstrate, even if each user's POIs individually is protected  by a location obfuscation mechanism, an adversary can still use the user mobility model to improve his estimation about user's true location and perform a localization attack.
\end{rem}

 \subsection{Utility of Service}
% We quantify the performance of LPPMs using privacy and utility metrics.
Inherently there is a utility loss associated with obfuscation mechanism. Here, we only require that the utility of service associated with the selected location obfuscation mechanism can be modeled as a real-valued and bounded MDP utility function $u$.  In particular, given a user's POI $s$, and the a location obfuscation mechanism $a$, $u(s,a)$ characterizes the service quality loss of reporting $a$ instead of the exact location of $s$. For the sake of completeness, through out the paper, we evaluate our privacy protection model using the cloaking-region method \cite{niu2013pseudo}. 

\subsection{User Mobility Model Illustration }\label{sec:example_illustration}
In this section, we aim to illustrate the proposed  MDP model based on a student mobility pattern on the campus of Notre Dame University. We consider a  student with a mobile device moving between the university libraries, and a LBS server that provides a service to the student.
 Figure \ref{fig:ND_MDP_example} shows  the Notre Dame university libraries map.
 We consider the libraries' locations as the student POIs and use the student movements to define the transition between the POIs. 
 Without loss of generality, we  assume the student uses the spatial cloaked region method \cite{gruteser2003anonymous} to obfuscate his true locations.
 Regions $A=\{a,b,c,d,e,f\}$ in the figure are the cloaked regions that are  precomputed by the anonymizer \cite{gruteser2003anonymous}. The student broadcasts the cloaked region associated with his source location to the LBS server and transits to the next location. 
 
Figure \ref{fig:ND_MDP_example} shows the constructed MDP model for the user mobility pattern. The states are the libraries' locations and the transition label shows the cloaking region information. The MDP transition function represents the probability of selecting the available cloaking regions at each POI. For instance, if the student at state $s$ uniformly selects a cloaking region $a \in A(s)$,  and moves to the state $s'$,  we define $T(s,a,s')= 1 \slash |A(s)|$. 
We assume the user starts from Hesburgh  library, labeled as $1$ in the  Figure \ref{fig:ND_MDP_example},  that sets the initial state distribution as $p_0=\left[1\ 0 \ 0 \ 0\ 0\ 0 \right]^T$. 
The quality loss function $u$ depends on the application and area of the cloaked region.
In this example, the student is querying for information in a specified area, and the obfuscation mechanism blurs the user location by increasing the area of retrieval that indicates the quality loss grows with the area of the cloaked regions. Therefore, a   quality loss metric for this setup can be defined by $u(s,a)=\frac{\text{Area}(a)}{\text{Area}(s)}$ \cite{ku2009privacy}, where $\text{Area}(.)$ is a function that computes  area of the specified region.

 \tikzstyle{state}=[rectangle,thick,draw=black!75, fill=black!20,minimum size=2mm]
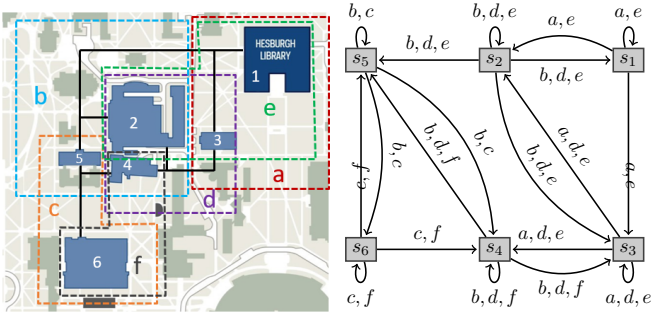
\begin{figure}[!t]
\centering
\begin{minipage}[b]{0.49\linewidth}
 \includegraphics[width=\linewidth]{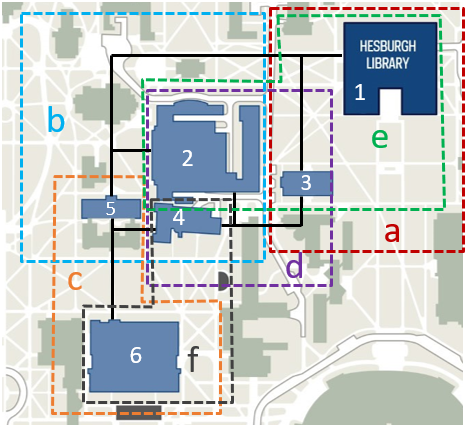}
%\caption*{a.}
%  \label{fig:ND_Campus_MAP_LBS}
\end{minipage}\hfill 
\begin{minipage}[b]{0.49\linewidth}
%\label{fig:ND_MDP_Model}
\resizebox{\linewidth}{!}{
\centering
\begin{tikzpicture}[scale=1,shorten >=1pt,node distance=2.4cm,on grid,auto, bend angle=20, thick, every node/.style={transform shape}] 
      \node[state,initial text=] (s_1)   {$s_1$};
	  \node[state] (s_2)   [left=of s_1,yshift=0cm,xshift=0cm] {$s_2$};
      \node[state] (s_3)   [below=of s_1,yshift=-1cm,xshift=0cm] {$s_3$};
      \node[state] (s_4)   [left=of s_3,yshift=0cm,xshift=0cm] {$s_4$};
      \node[state] (s_5)   [left=of s_2,yshift=0cm,xshift=0cm] {$s_5$};
      \node[state] (s_6)   [below=of s_5,yshift=-1cm,xshift=0cm] {$s_6$};
	\path[->]
	  (s_1) edge  [loop above]  node  {$a,e$} (s_1)
	  (s_1) edge  [bend right=30] node [ above] {$a,e$} (s_2)
	  (s_1) edge  [sloped] node [ ] {$a,e$} (s_3)
	  (s_2) edge  [loop above]  node  {$b,d,e$} (s_2)
      (s_2) edge  [] node [ below] {$b,d,e$} (s_1) 
      (s_2) edge  [sloped,bend right=30] node [above ] {$b,d,e$} (s_3) 
      (s_2) edge  [] node [ above] {$b,d,e$} (s_5)
     % (s_2) edge  [ ] node [ above]  [below]{$b,d,e$} (s_5) 
      (s_3) edge  [loop below]  node  {$a,d,e$} (s_3)
      (s_3) edge  [sloped] node [ above] {$a,d,e$} (s_2)
	  (s_3) edge  [] node [ above,  pos=0.7] {$a,d,e$} (s_4)
      (s_4) edge  [loop below]  node  {$b,d,f$} (s_4)
      (s_4) edge  [bend right=30] node [ below] {$b,d,f$} (s_3)
      %(s_4) edge  [] node [ right] {$b,d$} (s_2)
      (s_4) edge  [sloped] node [ above] {$b,d,f$} (s_5)
      
      (s_5) edge  [loop above]  node  {$b,c$} (s_5)
      (s_5) edge  [sloped, bend left=30] node [ ] {$b,c$} (s_4)
      (s_5) edge  [ sloped, bend left=20] node [above ] {$b,c$} (s_6)
      
      (s_6) edge  [loop below]  node  {$c,f$} (s_6)
      (s_6) edge  [sloped] node [ ] {$c,f$} (s_5)
      (s_6) edge  [ ] node [ ] {$c,f$} (s_4)
                    ;
			\end{tikzpicture}
			}
%\caption*{b.}
\end{minipage}
\caption{Left: Notre Dame university library map.
The blue marked regions are the university libraries,  the student walking paths are marked with black lines, and the spatial cloaking regions for each library are shown with dashed line. Right: MDP user mobility model.}
\label{fig:ND_MDP_example}
\end{figure}

%\section{LPPM Overview}
\section{Privacy Notation} \label{sec:privacy_model}
To evaluate an LPPM framework, it is important to define the location privacy metric and the adversary threat model. 
In the following, we assume the user's POIs contain confidential information that should be protected from an adversary, and accordingly, define the adversary threat model and the location privacy metric.

\subsection{Adversary Inferece Model} \label{subsection:Adversary_model}
We assume a LBS server is potentially an  adversary who can have access to the user's mobility  traces and other public information to construct the user's mobility  model \cite{shokri2011quantifying}. The perturbed location are also reported to the LBS server than  can be   visible to the  adversary. 
Therefore,  we consider an adversary who has a knowledge of the user's mobility model $M$,  and is capable of observing  the reported perturbed locations $A$, however, the users' true POIs or the  states $S$, are not observable to the adversary. 

The adversary objective is to perform a \textit{localization attack} that is to infer  the user presence at  his POIs \cite{shokri2011quantifying}.
In the localization attacks, the adversary obtains users obfuscated locations  that are generated probabilistically by the LPPM. 
Since for any observed obfuscated location $a \in A$, there are potentially many user's POIs that may have produced $a$,  the outcome of the adversary localization attack is a probability distribution over the user true POIs \cite{shokri2011quantifying}. Formally,  given an observation $a \in A$,  the adversary outcome for any $s \in S$ is in the form of a posterior distribution $\text{Prob}\left[s\mid a\right]$. We call this posterior distribution the adversary belief and we denote it by $b$ that is a function $b:S \to \left[0,1\right]$, such that $\sum_{s} b(s) =1$.
Furthermore, the adversary can have some background knowledge about users' mobility habit before starting his observation.  This side knowledge also can be encoded as a probability distribution over the state set \cite{ashbrook2003using} that is defined by $b_0:S \to \left[0,1\right]$, such that $\sum_{s} b_0(s) =1$.

The adversary initially at $t=0$, has a prior belief $b_0$, and at each time  instant $t$, when a set of perturbed locations $A$ are observed, the adversary updates its  belief  with Bayes rule by computing its posterior  distribution, given by:  
\begin{equation}  \label{eq:adversary_general_model}
b_{t+1} (q' \mid A) =  \sum\limits_{a \in A} p_{a_t} \sum\limits_{q \in S} T(q,a,q') b_t(q),
\end{equation}
where, $p_{a_t}=\text{Prob}\left[{a_t}=a\right]$ is the probability of $a$ at time $t$.

\subsection{Overview of Related Privacy Notations and Their Limitations}
Given the adversary threat model, LPPMs are designed and being evaluated based on an assumed location-privacy metric. %There are various metrics in the literature to quantify user's location privacy.
Selection of an effective privacy metric however, highly depends on the specification of the user's location privacy requirement. Here, we first provide an overview of related location privacy notations to motivate our proposed privacy metric.  For a comprehensive review of the privacy metrics, see  \cite{wagner2018technical}.

Throughout this section, we will be using the user mobility model in Figure \ref{fig:ND_MDP_example} as a running example. We use the adversary inference model in~\eqref{eq:adversary_general_model}, and assume
its prior belief $b_0$, is a uniform distribution. 
As it is mentioned earlier,  we consider an adversary with a localization attack model, and therefore, we focus on related privacy metrics that measure user's  \textit{presents-absence disclosure} to an adversary. Note that, we  will not evaluate the utility of service associated with these privacy metrics since our objective here is to only study the limitation of these metrics to capture the outcome of the adversary localization attacks. 

Let's assume the state $s_4$ is the user's secrete area and it is critical to ensure the adversary can not determine the user presence at this area.  
We denote the user's secrete POIs by $S_s=\{s_4\}$. We furthermore consider the user's LPPM as an obfuscation mechanism $f: S \to \Delta(A)$ that determines the random mapping between the user's actual POIs $S$, and the cloaked regions $A$. Note that, we assume the adversary can accurately measure the probability distribution of any obfuscated  location $a$, and hence, we set $p_{a}=\sum_s p(s)f(a|s)$ in the adversary inference model~\eqref{eq:adversary_general_model}.

\begin{figure*}
  \begin{minipage}{.32\textwidth}
    \includegraphics[width=\linewidth]{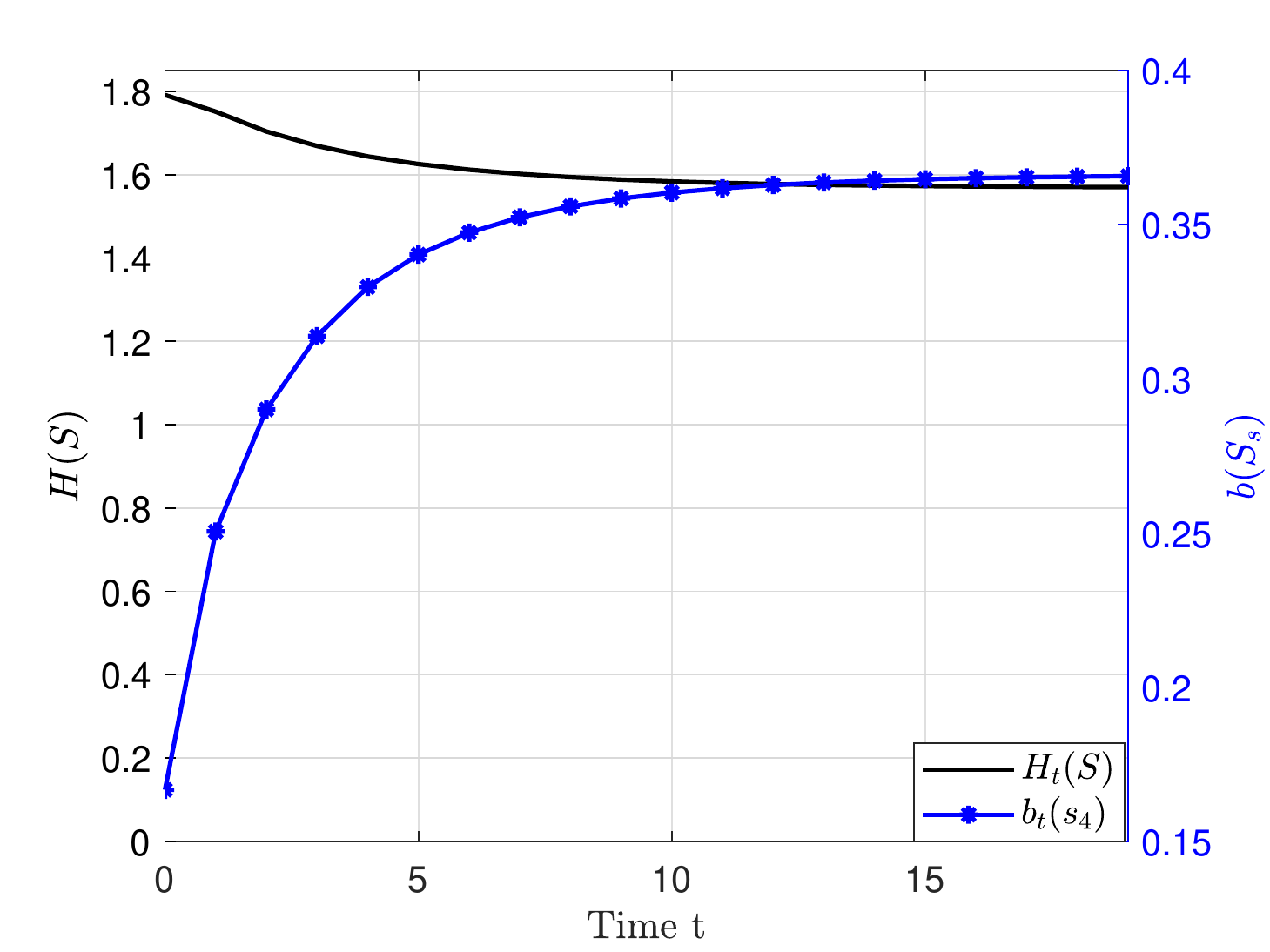}
     \caption{The entropy of the adversary belief in~\eqref{eq:posterior_entropy_with_time}. The left vertical axis represents entropy of the adversary posterior, and the right one is the adversary belief over the user secret area $S_s$.  } 
  \label{fig:example_belief_entropy}
  \end{minipage} \hfill
  \begin{minipage}{.33\textwidth}
    \includegraphics[width=\linewidth]{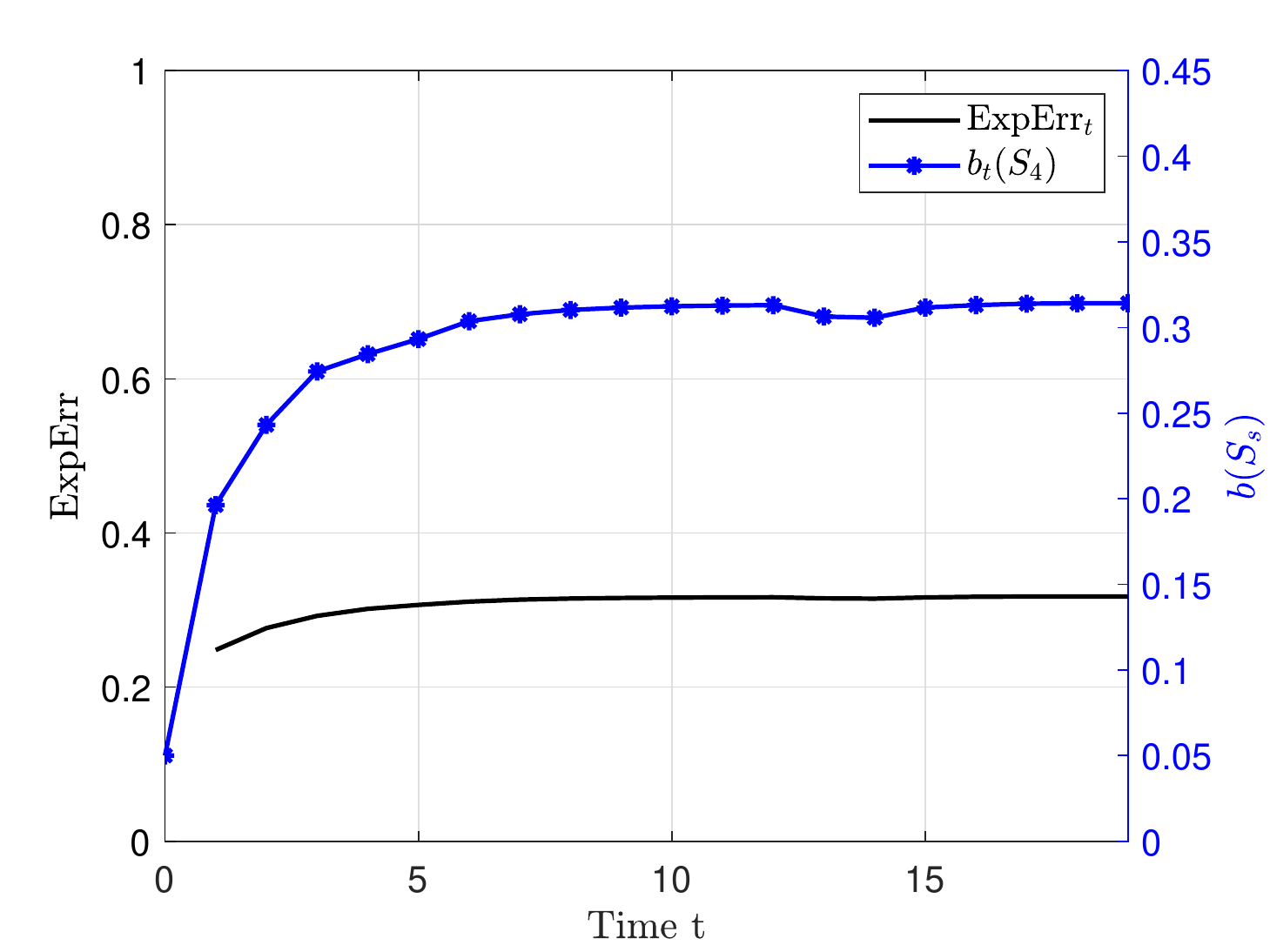}
    \caption{The adversary expected inference error~\eqref{eq:exp_inference_err_with_time}. The left vertical axis represents the minimum expected  adversary inference error, and the right one is the adversary belief over the user secret area $S_s$.  } 
  \label{fig:example_belief_exptErr}
  \end{minipage}  \hfill
 \begin{minipage}{.32\textwidth}
    \includegraphics[width=\linewidth]{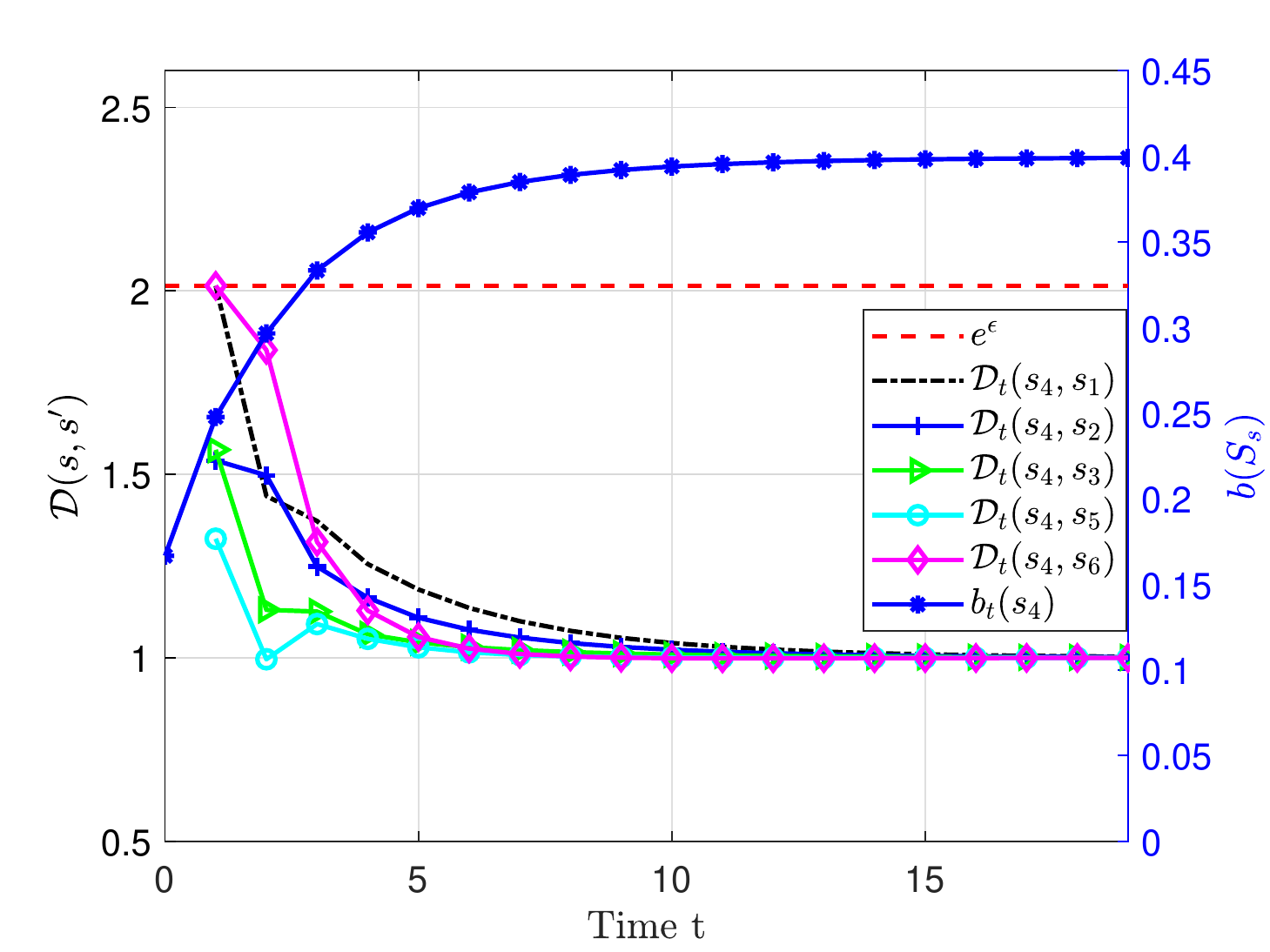}
     \caption{Adversary belief with $(D,\epsilon)-$privacy condition~\eqref{eq:diff_privac_belief_compact}. The dashed line is the $(D,\epsilon)-$privacy upper bound $e^\epsilon$.
  The left vertical axis represents $\mathcal{D}_t$  for $S_s$, and the right one is the adversary belief over $S_s$.  } 
  \label{fig:example_belief_DiffPrivacy}
  \end{minipage} 

\end{figure*}
\subsubsection{Entropy}
A common approach to define the user location privacy is to use \textit{entropy} to quantity the adversary localization attack's outcome \cite{ni2019anonymous}. 
%\cite{oya2017back,ni2019anonymous}. % ambiguity on the user localization is captured by \textit{certainty}. 
%Authors in \cite{ni2019anonymous}
%\cite{gruteser2003anonymous,bettini2005protecting,gedik2005location,xu2009feeling},
%proposed to measure  certainty of the adversary by entropy in k-anonymity model.
Entropy quantifies the uncertainty associated with predicting the value of a random variable. In the location privacy, it can be interpreted as how well the adversary can determine the user's position among other possible locations. % \cite{ni2019anonymous}.
%, and in $k$-anonymity model, identify a user among other $k$ users. 
More specifically, entropy of the adversary belief $b$,  is defined by:
% \begin{align}\label{eq:posterior_entropy}
%  H(S)= \sum_s \text{Prob}\left[s\mid a\right] \log (1 \slash \text{Prob}\left[s\mid a\right]).   
% \end{align}
\begin{align}\label{eq:posterior_entropy}
 H(S)= \sum_s b(s) \log (1 \slash b(s)).   
\end{align}
The maximum adversary uncertainty is when $b$ is a uniform distribution that implies, the higher entropy is,  the lower certainty for the attack's outcome.  The absolute value of entropy however does not necessarily indicates how accurate the adversary attack's outcome is. For instance, an adversary can identify the user current location with high probability, but the remaining low probability of other possible locations still  results a high value of entropy \cite{toth2004measuring}. To further illustrate this limitation,  we compute the obfuscation mechanism  $f$ that maximizes the adversary entropy $H(S)$. Formally,
 \begin{align} \label{eq:posterior_entropy_with_time}
    \max_{f_t} H_t(S) \quad t \ge 1.
\end{align}
where $H_t(S)= \sum_s b_t(s) \log (1 \slash b_t(s))$. 
Figure \ref{fig:example_belief_entropy} demonstrates the adversary entropy over the user's POIs $S$, and the adversary belief over the user's secrete location $s_4$. 
As the figure illustrates, although the absolute value of $H(S)$ did not change drastically, the adversary becomes significantly more confident about the user presence at $s_4$. 
Note that, for the MDP in Figure \ref{fig:ND_MDP_example}, the maximum adversary entropy is $\log(|S|)= 1.79$. The slight reduction of $H(S)$ is due to the constraints imposed by the user's mobility model on the obfuscation mechanism $f$.

 \subsubsection{Inference Error}
Privacy metrics based on \textit{expected inference error} are widely used to characterize location privacy \cite{shokri2011quantifying}.
These approaches quantify how close the attack’s outcome is to the true user's locations.
%\textit{Correctness} quantifies how close the attack’s outcome is to the truth. Authors in \cite{shokri2011quantifying} propose to use \textit{expected estimation error} of the adversary for this type of metric.
Given the attack's outcome $b$, this privacy metric is defined by:
% \begin{align}\label{eq:exp_inference_err}
%  \text{ExpErr}= \min_{s} \sum_{\hat{s} \in S} \text{Prob}\left[\hat{s} \mid o \right] d_q(\hat{s} , s) 
% \end{align}
\begin{align}\label{eq:exp_inference_err}
 \text{ExpErr}= \min_{\hat{s}} \sum_{s \in S} b({s}) d_q( s,\hat{s}), 
\end{align}
 where $d_q(.)$ can be Hamming distance or Euclidean distance between $\hat{s}$ and $s$, or any other metric that captures the privacy loss. Here,  $\hat{s}$ is the estimated user's location, and $s$ is the user's true location. In this privacy metric, the high estimation error correlates with low location privacy, which indicates the adversary objective is to minimize the expected error of any attack's outcome. The objective LPPM for this setup is to maximize the adversary inference error given the user constraint on the service quality loss \cite{shokri2012protecting}. This privacy metric defines user privacy as a global performance metric that is averaged over all locations. Hence, this metric does not explicitly characterize the the user privacy at a certain location which can be critical for the localization attacks. 
 Similarly, we design an LPPM to illustrate this limitation. 
 %For sake of clarity, let's consider the $d_q(s,\hat{s})$ as a Hamming distance in which $d_q(s,\hat{s})=0$ if $s=\hat{s}$, and  $d_q(s,\hat{s})=1$ otherwise. Then it is easy to see the expected inference error~\eqref{eq:exp_inference_err} becomes $\text{ExpErr}= \max_{\hat{s}} b(\hat{s})$, that intuitively implies the optimal localization attack that estimates the user true location is the one having the maximum adversary belief.    
 Let's consider  $d_q(s,\hat{s})$ as a Euclidean distance between centroid of area $s$ and $s'$.
 We compute the following obfuscation mechanism that   maximizes~\eqref{eq:exp_inference_err} for $t\ge 1$:
 \begin{align} \label{eq:exp_inference_err_with_time}
     \text{ExpErr}_t= \max_{f_t} \min_s \sum_{\hat{s} \in S} b_t(\hat{s}) d_q(\hat{s} , s).
 \end{align}
 Figure \ref{fig:example_belief_exptErr} demonstrates how the  expected inference error has evolved with respect to the adversary belief over the user secret location $b_t(s_4)$. As the simulation result is illustrated,  the LPPM could maintain the  expected inference error of the adversary close to the initial error, but $b_t(s_4)$ has increased over time that indicates the adversary becomes more certain about the user presence at $s_4$.

\subsubsection{Differential Privacy}
Differential privacy is a rigorous mathematical framework that provides a provable privacy guarantee to protect individual data in a database.  This is achieved by adding controlled noise to the query outcome such that,  the presence (or absence) of an individual in the database will have a negligible impact on the perturbed reported answer. 
Differential privacy techniques have been used in the context of LBS by considering the user's location as the sensitive information and using the location obfuscation mechanisms to produce the query noise. These obfuscation mechanisms are designed  such that  changing from one location to another nearby location makes the probability distribution of the reported locations to only change to a  certain extent \cite{elsalamouny2016differential}.

Depending on the exact notion of indistinguishability between the user's locations, several differential privacy models have been introduced. 
In $\epsilon-$geo-indistinguishability models   \cite{andres2012geo}, the indistinguishably between two adjacent locations proportionally increases with the distance between them. 
Formally, the obfuscation mechanism $f$ satisfies $\epsilon-$geo-indistinguishability if $ {f(a|s)}\slash {f(a|s')} \le \exp{(\epsilon d_q(s,s'))} $. 
This property means, the (log of) the ratio between the probability of reporting the obfuscated location  $a$, at any pair of user's locations, is always bounded by $\epsilon d_q$.

Another popular model based on the differential privacy is $(D,\epsilon)-$location privacy \cite{elsalamouny2016differential}. In this model, indistinguishably is defined for any arbitrary locations within a distance of a predefined value $D$.  Formally, $(D,\epsilon)-$location privacy holds for the mechanism $f$ if ${f(a|s)}\slash{f(a|s')} \le \exp({\epsilon})$ for all $s,s' \in S$ such that $d_q(s,s') \le D.$
This characterization implies the (log of) the ratio of between the probability of reporting the perturbed location $a$ in any two adjacent locations is at most $\epsilon$. These two privacy notation are  correlated \cite{elsalamouny2016differential}, and  have been extended  to characterize other type of location indistinguishably \cite{chatzikokolakis2017methods}. However, our objective here is to evaluate the differential-privacy based metrics for the described localization attack. Therefore, for the sake of simplicity and clarity,  we focus our analysis on the specific model  of $(D,\epsilon)-$privacy. Particularly, we select $D$ such that $d_q(s,s')\le D$ holds for any pair of states $s,s'\in S$.
Differential privacy in general and $(D,\epsilon)-$privacy, in particular, are solely described by the obfuscation mechanism, and hence, they can guarantee the user privacy independent of the adversary prior knowledge. However,  \cite{elsalamouny2016differential} has shown that it might not be possible to achieve $(D,\epsilon)-$privacy for any adversary prior knowledge while receiving a reasonable quality of service. Therefore, it is important to characterize $(D,\epsilon)-$privacy as the adversary knowledge gain. More precisely,  let's consider $b_{t+1}$ as the adversary posterior and $b_{t}$ as the prior belief. The obfuscation mechanism $f$ satisfies $(D,\epsilon)-$privacy at time $t$, if and if the adversary posterior and prior belief  satisfy \cite{elsalamouny2016differential}:
\begin{align} \label{eq:diff_privac_belief}
    \frac{b_{t+1}(s)}{b_{t+1}(s')} \le e^\epsilon   \frac{b_{t}(s)}{b_{t}(s')} \quad \forall s,s' \in S.
\end{align}
As~\eqref{eq:diff_privac_belief} implies, $(D,\epsilon)-$privacy provides location privacy by restricting the relative gain of the adversary knowledge  about the user's true location. However, this model may not be able to quantify how close the adversary belief gets to the user's true location over time. We illustrate  this limitation through the student mobility model in Figure \ref{fig:ND_MDP_example}. For any $s,s' \in S$, and time $t$, let's rewrite~\eqref{eq:diff_privac_belief} as:
\begin{align}\label{eq:diff_privac_belief_compact}
    \mathcal{D}_t({s},{s'})=\frac{ b_{t+1}(s)b_{t}(s') }{ b_{t+1}(s')b_{t}(s)} \le e^\epsilon,
\end{align}
where $b_{t+1}$ can be obtained from the adversary inference model~\eqref{eq:adversary_general_model}. 
we set $\epsilon=0.7$, and   find $f$ that makes the adversary belief $b_{t+1}$, to  satisfy~\eqref{eq:diff_privac_belief_compact} at any time $t$.
Figure \ref{fig:example_belief_DiffPrivacy} demonstrates how $\mathcal{D}_t$ evolves  for the secrete location $s_4$. As it is depicted in the figure,  the  relative adversary knowledge gain over $s_4$,  captured by $\mathcal{D}_t(s_4,s)$ for any $s\in \{S - s_4\}$, meets the $(D,\epsilon)-$privacy criteria~\eqref{eq:diff_privac_belief_compact}. However, the absolute value of the adversary belief over the secrete location, $b_t(s_4)$, is kept increasing over time. Hence, although the relative adversary knowledge gain is upper bounded by $e^\epsilon$, and is safe in a sense of $(D,\epsilon)-$privacy, the adversary can still combine the reported obfuscated locations and the user mobility model to further improve his knowledge about the user secrete location.

Roughly speaking, the underlying limitation of these metrics to capture the adversary knowledge over a certain location relies on quantifying the privacy based on the location indistinguishably. 
Motivated by probabilistic current state opacity (CSO) metric \cite{Wu2020}, we propose to define the user location privacy as the absolute value of the adversary belief over the user's secrete area.

\subsection{Location Privacy in Belief Space}
To quantify the level of protection offered by an LPPM against the adversary, we propose to use a probabilistic current state opacity metric (CSO). In the CSO framework, the goal is to protect the system states that contain sensitive information, namely \textit{secret} states, against an external observer or an adversary  \cite{Wu2020}.
The adversary maintains a belief
over the system secrets through Bayesian inference, and the system is considered CSO safe if the adversary confidence that a secret has been observed is bounded. 
In the LPPM framework, in analogs to secret states, the users may only need to protect a subset of POIs which contains users semantic information that can be used by an adversary for profiling or other attacks. Examples include health clinics, places with religious significance, etc. \cite{ashbrook2003using}. Therefore, we propose to use the CSO definition to quantify users' POIs privacy against an adversary with the intention of localization attack.
Formally, suppose there is a subset of states in the MDP $M$, representing sensitive POIs of a user, $S_{s} \subset S$, that LPPM would like to conceal from the adversary. Note that, for $S_{s} =S$, the problem is trivial.

Our CSO inspired location privacy metric is required that the adversary belief of $M$ being in the secrete states $S_{s}$, is upper bounded by a constant $\epsilon \in \left[0,1\right]$. 
We denote this privacy notation as a $\epsilon-$\textit{privacy} and formalize it in the following definition.

\begin{definition}
 Given user mobility pattern modeled by a MDP $M=(S,A,T,p_0,u)$, set of secrete  POIs, $S_{s} \subset S$, and $\epsilon > 0$, $M$ is $\epsilon-$private under privacy-preserving policy $\mu$, if
  \begin{equation}\label{eq:CSO_opacity}
\sum_{s \in S_{s}} {b_t(s)}   \le \epsilon \quad t> 0.
 \end{equation}
\end{definition}
 This notation of privacy characterizes inference evaluation of the adversary observation over the user's current true location for all time. In this setup, the location privacy of the user is preserved if the adversary confident is lower than some desired threshold. 
 Note that, in general, we cannot prevent the adversary from having prior information about the user's mobility pattern, and therefore we assume the system $M$ is $\epsilon$-private at $t=0$. However, our goal is to control the additional information that the adversary obtains by observing the released location information by the user.

\section{Location-Privacy Protection Problem} \label{sec:LPP_problem}
In our framework, the user location protection mechanism against the adversary localization attack is based on the randomization of the obfuscation mechanisms that is provided by a \textit{privacy-preserving policy}.
%We propose to design a  policy that randomizes the available obfuscation mechanisms at each of the user's POIs.
\subsection{Privacy-preserving Policy} 
A \textit{stationary} stochastic privacy-preserving is a function $\mu: S \times A \to \left[0,1\right]$, that assigns a probability distribution over the actions of each state $s \in S$. 
We denote  the policy  \textit{decision matrix} as $\mathbf{\mu} \in \mathbb{R}^{n,m}$, where each element of $\mathbf{\mu}$ in $s$th row and $a$th column  is defined by $\mu(s,a)$.
MDP $M$ is said to follow a  policy $\mu$, if for any state $s$, it draws an action $a$ from $\mu$.
Therefore, for any $s\in S$, and $a \in A$, we have $\mu(s,a)=\text{Prob}\left[a \mid s\right]$.
%A policy is call stationary if $\mu_t=\mu$ for all $t$.
%Let's define set $\pi=\left[\mu_0,\dots,\mu_{N-1} \right]$ as a policy for $N$ epochs.
Note that  $\mu$ induces a Markov chain (MC) $\text{MC}=(S,M_{\mu},p_0 )$ \cite{puterman2014markov} with transition matrix $M_{\mu} \in \mathbb{R}^{n,n}$, where its elements for any $s,s' \in S$ are:
\begin{equation} \label{eq:M_mu}
    M_{\mu}(s,s')=\sum\limits_{a \in A} \mu(s,a) T(s,a,s').
\end{equation}
Let $T^a \in \mathbb{R}^{n,n}$ be a row stochastic matrix having elements $ T(s,a,s')$ for any $s,s' \in S$. The induced MC transition matrix $ M_{\mu}$ in \eqref{eq:M_mu}, can be written as \cite{el2015finite}:
\begin{equation} \label{eq:M_mu_compact}
    M_{\mu}=\sum\limits_{a \in A}  T^a \odot \left( (\mathbf{\mu} e_a) \mathbf{1}^T\right),
\end{equation}
where $\mathbf{1}$ is a column vector of all ones, $e_a$ is a vector with the $a$th element as $1$ and the rest are all zeroes, and $\odot$ is a element-wise \textit{Hadamard} product.
For each state $s\in S$, let's denote $p_t(s)=\text{Prob}\left[s_t=s\right]$ as the  probability of the system being at state $s$ at time $t$, and let $p_t \in \mathbb{R}^n$ be the vector of  $p_t(s)$ for all $s \in S$. The state probability distribution of $\text{MC}$ evolves according to the following dynamic equation:
\begin{equation} \label{eq:stationary_dist}
   p_{t+1}=M^T_{\mu} p_t.
\end{equation}
\begin{rem}
 Considering stationary privacy-preserving policy,  helps the LPPM to avoid recomputing the policy for every change of the adversary knowledge  that dramatically reduces the computational overhead on the user's mobile device. 
\end{rem}

When the policy $\mu$ is stationary, the induced MC transition matrix, $M_\mu$, is time-independent.
For finite state MC with transition matrix $M_{\mu}$, the \textit{stationary distribution}, $p_{\infty}$, is a   row probability vector satisfying: 
\begin{align}
    p_{\infty}=M_\mu p_{\infty},
\end{align}
that implies if after some time $t'$, $p_{t'}= p_{\infty}$, then for all $t' \ge t$, we have $p_{t'}= p_{\infty}$.
 $M_\mu$ is called \textit{ergodic}, if there exists a unique invariant and strictly positive stationary distribution $p_{\infty}$, such that  independent of the initial distribution,
$p_0$, the state probability distribution, $p_t$, converges to $p_{\infty}$, i.e., $\lim_{t \to \infty} p_t = p_{\infty}$ \cite[\S8.3]{puterman2014markov}.

\begin{definition}\label{def:unichain}
 A MDP $M$ is called \textit{unichain}, if for each policy $\mu$, the Markov chain induced by $\mu$ is {ergodic} \cite{puterman2014markov}.
\end{definition}
Here, we assume the MDP $M$ representing the user mobility pattern  is a unichain MDP.
This indicates that for any policy $\mu$, the induced MC, $M_\mu$, has  a single \textit{recurrent class} plus a possibly empty set of  \textit{transient states} (see \cite[\S8.3]{puterman2014markov} for an elementary exposition of classification of MDP).
A recurrent state in the user mobility model is translated as a user POI $s$, which is accessible from all the POIs that are, in turn, accessible from the $s$.
Therefore, restricting $M_\mu$ to a single recurrent class of POIs, implies the user can travel to any POI from any other POIs which we believe is not a restrictive assumption on the user mobility model.

\subsection{Performance Metrics}
Performance metric measures overall efficiency of the framework for a given LPPM and  quality loss model. 
% Performance metrics are orthogonal and do not participate to the privacy and utility trade-off, but still are important because they impact the usability of LPPMs.
When the privacy preserving policy frequently makes a  decision in a long-run, it might be more preferable to compare LPPMs on the basis of their {average expected quality loss}.  We therefore consider  \textit{average quality loss} criterion as as a performance metric \cite{puterman2014markov}. Formally, for any stationary policy $\mathbf{\mu}$, and initial distribution $p_0$, the expected total  quality loss is defined as follows:
\begin{align} \label{eq:v_mu}
    v_{\mu}= \lim_{N\to \infty} \frac{1}{N} \mathbb{E}^{\mu}_{p_0}\left[\sum\limits_{t=1}^{N}u(s_t,a_t)\right],
\end{align}
This expression may be undefined when the limit does not exist, however it known that 
for unichain MDP, the existence of this limit is guaranteed \cite[\S8.2]{puterman2014markov}. 
Here, LPPM objective is to minimize the  overall quality loss while the user privacy is preserved according to the privacy condition~\eqref{eq:CSO_opacity}.

\subsection{Problem Formulation}
One the objective of the LPPM  is to find an optimal privacy-preserving policy $\mathbf{\mu}^*$ that minimizes  $v_{\mu}$.  
Formally, $v^*_{\mu}=\min_{\mathbf{\mu}} v_{\mu} $, and the optimal privacy-preserving policy is $\mu^*= \argmin_{\mathbf{\mu}} v_{\mu}$. If there is no constraint on the  states and actions, this problem becomes a well-studied unconstrained MDP planning, and it can be solved by value-iteration algorithms \cite{puterman2014markov}. 
However, here we have location privacy constraints that have to be considered in the  policy synthesis. 
The problem of privacy-preserving optimal policy synthesis is given below.

\begin{problem}\label{prob:alltime_CSO_optimal_policy}
Given a unichain MDP $M=(S,A,T,p_0,u)$ with average total quality loss criterion $v_\mu$, synthesize an optimal privacy-preserving  policy $\mathbf{\mu}$ that  minimizes  $v_\mu$, and  ensures $M$ is $\epsilon-$private  for  $t \ge 0$.
\end{problem}

\section{Optimal Privacy-preserving Policy Synthesis} \label{sec:all_time_LPPM}
In this section, we study the design of privacy-preserving  optimal  policy for Problem \ref{prob:alltime_CSO_optimal_policy}. Our underlying idea is to formulate the policy synthesis problem into a  linear programming (LP) problem.
Given a MDP $M$ controlled by the policy $\mathbf{\mu}$, the average quality loss metric~\eqref{eq:v_mu} can be defined by 
$ v_{\mu}= \sum_{s,a} \mu(s,a) p_{\infty}(s)u(s,a)$ \cite[\S8.4]{puterman2014markov}, where $p_{\infty}$ is the unique stationary distribution of $M_\mu$.
Let's  denote joint state-action distribution as $\theta(s,a)=p_{\infty}(s) \mu(s,a)$, 
% \begin{align} \label{eq:state-action}
%   \theta(s,a)=p_{\infty}(s) \mu(s,a), 
% \end{align}
and in a matrix form, by $\mathbf{\theta} \in \mathbb{R}^{n,m}$, where $\sum_{s,a} \theta(s,a)=1$. 
The expected total quality loss metric~\eqref{eq:v_mu} then can be re-written as: 
\begin{align} \label{eq:v_mu_average_simple}
    v_{\mu}= \sum\limits_{s,a}  \theta(s,a) u(s,a),%=\mathbf{1}^T(p^T_{sa}R)\mathbf{1},
\end{align}
which is  a linear function of  $\mathbf{\theta}$ \cite[\S8.4]{puterman2014markov}.
%Maximizing \eqref{eq:v_mu_average_simple} without any CSO constrains, is a well studied problem (see e.g., \cite[\S8.2]{puterman2014markov}),  however, here our aim is to maximize $v_\mu$ while  
The stationary distribution requirement~\eqref{eq:stationary_dist} for $M_\mu$ similarly can be driven based on $\theta$, defined by:
\begin{align} \label{eq:state-action_stionary_dist_const}
    \sum_a \theta(s',a) = \sum_{s,a} \theta(s',a) T(s,a,s'), \quad \forall s' \in S.
\end{align}

The adversary belief also can be modeled as function of state-action distribution $\mathbf{\theta}$.  We considered an adversary who has access to the history of released user's locations, and therefore estimates the  distribution of the observed actions as $\hat{p}_{a_t}= \sum_s \mu(s_t,a_t)p(s_t)$.
The adversary then updates its belief  for any state $q'\in S$, based on a Bayesian
rule \cite{Bo8392381}, that is defined by the following equation:
\begin{equation}  \label{eq:adversary_model_with_p(a)}
 b_{t+1} (q'\mid A) =  \sum\limits_{a \in A} \sum\limits_{s \in S} \sum\limits_{q \in S} p_t(s)\mu(s,a)   T(q,a,q') b_t(q).
\end{equation}

However, since we consider a stationary  policy $\mu$, the user released locations are solely the function of the user's POIs, and therefore the adversary estimation of ${p}_t(a)$ in~\eqref{eq:adversary_general_model}, denoted by $\hat{p}(a)$, is time independent. Hence, we have $\hat{p}(a)= \sum_{s} p_\infty (s) \mu(s,a)$, that makes the adversary update dynamic model to take the following form:
\begin{align} \label{eq:adversary_model_unichain}
 b_{t+1} (q'\mid A) =  \sum\limits_{a \in A} \sum\limits_{s \in S} \sum\limits_{q \in S} \theta(s,a)   T(q,a,q') b_t(q).
\end{align}
Let's $b_t \in \mathbb{R}^n$ be the vector of adversary belief over state set $S$ at time instant $t$.  
The adversary belief $b_t$ can be seen as a state probability distribution of a MC with a dynamic expressed  in the following Lemma.
\begin{lemma} \label{lem:bt_dynamic}
The adversary belief $b_t$ evolves according to the following MC dynamic:
\begin{align} \label{eq:bt_dynamic}
  b_{t+1}=M^T_{adv} b_t, \quad t=0,1,\dots
\end{align}
where the adversary transition matrix is 
\begin{align} \label{eq:M_p_infty,mu_unichain}
M_{adv}=\sum\limits_{a} T^a \odot \left(  (\mathbf{1}^T \mathbf{\theta} e_a)  \mathbf{1} \mathbf{1}^T \right). 
\end{align}
\end{lemma}
\begin{proof}
 For all $q,q' \in S$, let's define the elements of $M_{adv}$ as: %$M_{adv}(q,q')=\sum\limits_{a} T^a_{qq'} \odot \left(  (\mathbf{1}^T \theta(q,a) e_a)  \mathbf{1} \mathbf{1}^T \right)$
 \begin{align*}% \label{eq:M_p_infty,mu_unichain}
M_{adv}(q,q')=\sum\limits_{a} T^a_{qq'} \odot \left(  (\mathbf{1}^T \theta(q,a) e_a)  \mathbf{1} \mathbf{1}^T \right),
 \end{align*}
then \eqref{eq:bt_dynamic} follows from definition of $b_{t+1}$ given in~\eqref{eq:adversary_model_unichain}. Furthermore, all the elements of $M_{adv}$ are non-zero, and $\sum_{q'} M_{adv}(q,q')=1$, and therefore $M_{adv}$ is a stochastic matrix and hence \eqref{eq:bt_dynamic} is a Markov chain model. 
\end{proof}

In order to enforce the $\epsilon-$privacy requirement on  $M$ with stationary policy $\mu$, we first find the necessary and sufficient conditions that makes $M$ to be $\epsilon-$private.
Let's denote $\mathbf{0}$ as  a zero matrix of appropriate dimension, and $\beta > (\ge) G$ as $\beta_{i,j} > (\ge) G_{i,j}$ for all $i,j$ elements of matrices $\beta$, and $G$. 

\begin{thm}
  For every $b_{0}$  satisfying $A_s b_0 \le \epsilon$, $M$ is $\epsilon-$private under policy $\mu$ for $t\ge 0$, if and if there exists $z \in \mathbb{R}$ and $\beta \in \mathbb{R}^{1,n}$ such that:
\begin{align} \label{eq:CSO_NSCond_them}
  &    -\mathbf{1}^T\epsilon z + A_s z -  A_s  M^T_{adv } -\beta \ge - \mathbf{1}^T \epsilon, \quad \beta \ge \mathbf{0}, z \ge 0.
 \end{align}
\end{thm}

\begin{proof}
We first find the necessary and suffocate conditions on $M_{adv}$ that ensures if $A_s b_0 \le \epsilon$, then $A_s b_t \le \epsilon$ for all $t >0$.
This  requirement can be written as:
\begin{align} \label{eq:CSO_verf_optimization_cond}
  &  \max \limits_{b \in \Delta(S), A_s b \le \epsilon} A_s M^T_{adv} b \le \epsilon.
\end{align}

This condition implies that any prior belief $b$ that respects the $\epsilon-$privacy condition, it makes  the posterior belief  to satisfy the $\epsilon-$privacy requirement.  In a standard form, it can be written as:  
 \begin{align}
  &\min \limits_{b} c \cdot b , \text{ s.t. }  \hat{A}_s b \ge d, b \ge 0,
 \end{align}
 
where $c=-A_s M^T_{adv} $, $d=\begin{bmatrix}-\epsilon & -{1} & {1} \end{bmatrix}^T$, and  $ \hat{A}_s=\begin{bmatrix} -A^T_s &  -\mathbf{1} & \mathbf{1} \end{bmatrix}^T$.
In order to evaluate the feasibility of prime optimization problems, let's first introduce  slack variable $\lambda \ge 0$ as $A_s b + \lambda = \epsilon$. 
Since $\epsilon > 0$, then we can always find $b \ge \mathbf{0}$, and $\lambda \ge 0$ such that $ \mathbf{1}^T b ={1}, A_s b + \lambda = \epsilon$, has a solution. Hence,  the prime optimization problem is feasible, and dual of this optimization problem takes the form: % with  $ \begin{bmatrix} z &y_1 & y_2 \end{bmatrix}$ is % add dimension 
\begin{align}
  &  \max \limits_{z,y} - \epsilon z - y_1 + y_2\\
  & \text{s.t. } -A_s z - \mathbf{1}^T y_1 + \mathbf{1}^T y_2  \le - A_s M^T_{adv},\nonumber \\ 
  & z,y_1 \ge 0 ,y_2\ge 0.\nonumber 
\end{align}

In a similar way, let's denote  slack variable $\beta \in \mathbb{R}^{1,n}$,   $\beta \ge \mathbf{0}$ that results:
     \begin{align} \label{eq:ver_opc_eq_const}
  -A_s z - \mathbf{1}^T y_1 + \mathbf{1}^T y_2  + A_s M^T_{adv} +\beta =  \mathbf{0}.
\end{align}
 It is clear that the quality constraints has a solution for $ \beta\ge \mathbf{0}$, and $z,y_1,y_2\ge 0$ that implies the strong duality holds,  i.e., $-A_s M_{adv} b =  - \epsilon z - y_1 +y_2$. 
 Let's denote $y=y_2-y_1$, then from  equality condition \eqref{eq:ver_opc_eq_const}, we have
$\mathbf{1}^Ty = A_s z -  A_s M^T_{adv} -\beta$, that implies dual of the optimization  takes the form: 
\begin{align} \label{eq:dual_prob_cso_verf_temp}
  &  \max \limits_{z,y,\beta} \quad  - \epsilon z + y\\ \label{eq:dual_prob_cso_verf_temp_const}
  & \text{s.t. }\mathbf{1}^Ty = A_s z -  A_s M^T_{adv} -\beta,  \\  
  & \beta \ge \mathbf{0},  z \ge 0,y \text{ unconstrained}. \nonumber
%  & \beta \ge \mathbf{0}, z \ge 0,y > 0.
\end{align}
 This is  equivalent to set of  LP  problems, given by: 
\begin{align} \label{eq:dual_prob_cso_verf}
  &  \max \limits_{z \ge 0,\beta \ge \mathbf{0}}   \left[-\mathbf{1}^T\epsilon z + A_s z -  A_s M^T_{adv} -\beta\right] e_i,  %\\
%  & \text{s.t. } A_s z -  A_s M^T_{adv} -\beta \ge \mathbf{0},\quad \beta \ge \mathbf{0}, z \ge 0.\nonumber
\end{align}
where $e_i$  is a vector with the $i$th element as $1$ and the rest are all zeroes, and $i \in \{1,\dots,n\}$. Here by equivalent we mean if there exists a feasible solution that maximizes~\eqref{eq:dual_prob_cso_verf_temp}, that must also maximize all the linear programming problems in~\eqref{eq:dual_prob_cso_verf}. Note that~\eqref{eq:dual_prob_cso_verf} 
is obtained by multiplying $\mathbf{1}^T$ to~\eqref{eq:dual_prob_cso_verf_temp}, and substituting  $\mathbf{1}^T y$ with~\eqref{eq:dual_prob_cso_verf_temp_const}.  
Since strong duality holds and duality gap is zero, the necessary and sufficient condition is using the LP problem  expressed in~\eqref{eq:CSO_verf_optimization_cond}. Hence the optimal solution of~\eqref{eq:dual_prob_cso_verf} must satisfy:
\begin{align} \label{eq:necc_suff_cond_CSO_verf}
  &    -\mathbf{1}^T\epsilon z + A_s z -  A_s  M^T_{adv} -\beta \ge - \mathbf{1}^T \epsilon, \quad \beta \ge \mathbf{0}, z \ge 0.
 \end{align}

Therefore, if there exists $\beta \ge \mathbf{0}$, $z \ge 0$ satisfying~\eqref{eq:necc_suff_cond_CSO_verf},  that must be the solution of  dual problem~\eqref{eq:dual_prob_cso_verf}, and since strong duality holds, it is the solution of~\eqref{eq:CSO_verf_optimization_cond} that guarantees $A_s b_t \le \epsilon$ for  $t>0$, and hence $M$ is $\epsilon-$private . 
Note that, a partially similar scheme of the proof was used by \cite{accikmecse2015convex} for the safety control Markov Chain. Here, we have polytopic constraints which is a different problem formulation from  \cite{accikmecse2015convex}. 
 This proves the necessary and sufficient conditions. 
\end{proof}

Inequality~\eqref{eq:CSO_NSCond_them} characterizes a condition over $ M_{adv}$ that if satisfies implies:  if  MDP $M$ is  initially $\epsilon-$private at $t=0$, it stays $\epsilon-$private forever. 
We encode this condition as a constraint over policy $\mu$ in linear  programming that minimize the average quality loss metric $v_\mu$ given in~\eqref{eq:v_mu_average_simple}. 
The following LP solves Problem \ref{prob:alltime_CSO_optimal_policy}. 
\begin{align} \label{eq:LP_policy_design}
  &  \min \limits_{ \theta, M_{adv}}  v_\mu\\
  & \text{s.t. } \eqref{eq:CSO_NSCond_them}, \eqref{eq:M_p_infty,mu_unichain}, \eqref{eq:state-action_stionary_dist_const}, \sum_{s,a} \theta(s,a)=1,  \theta \ge \mathbf{0}. \nonumber
\end{align}
Note that, we have redefined our objective function $v_\mu$, and the adversary inference model $M_{adv}$  based on the state-action variable $\theta$. If $\theta^*$ is an optimal solution that minimizes   \eqref{eq:LP_policy_design}, then for any $s\in S$, and $a \in A$, it  induces a  stationary state distribution $p^*_{\infty}(s)= \sum_a \theta^*(s,a)$, and  a stationary optimal  policy given by $\mu^*(s,a)= \theta^*(s,a)\slash p^*_{\infty}(s)$. 

\begin{rem}
Our privacy protection framework is designed for an individual mobile user, and the privacy-preserving policy is designed to be performed at the user’s mobile device without the collaboration of a trusted third party or other users. This property of our framework is an important advantage compared to other types of solutions that rely on the intervention of a trusted centralized service, such as spatial $k-$anonymity \cite{gruteser2003anonymous}.
\end{rem}

\subsection{LPPM with unsafe adversary prior belief}
To this end, we have assumed the adversary prior belief over the secrete  user's POIs, does not violate the location privacy, i.e., $A_s b_0 \le \epsilon$.  This assumption may not be valid if the adversary has a strong background knowledge about the user mobility model, or the LPPM designer selects a tight upper bound over adversary belief that makes $b_0$ unsafe \cite{christin2016privacy}. 
Here, we aim to design a LPPM that drives the adversary belief to a $\epsilon-$private region even if her prior belief does not meet the  $\epsilon-$privacy criteria.

Our LPPM design process for this setup relies on the invariant property of ergodic Markov chain that guarantees the state probability distribution converges to a  stationary distribution independent of the initial distribution of the states. Hence, if the adversary inference model is an ergodic MC, it guarantees the adversary belief over the user's secrete  POIs, will converge to a unique stationary distribution. The LPPM objective then is to force the adversary belief  to converge to a stationary distribution that respects  $\epsilon-$privacy requirement.
\begin{lemma}
If MDP $M$ is unichain, then for any  policy $\mu$, and any stationary distribution $p_{\infty}$, respecting~\eqref{eq:stationary_dist}, $M_{adv}$ is an ergodic Markov chain.
\end{lemma}
%, and define $\Tilde{P_\mu}=\Tilde{\mu}(a) T$
\begin{proof}
Let's consider a stationary policy $\Tilde{\mu}$, and denote $\Tilde{\mu}(a)= \sum_s p_{\infty}(s) \mu(s,a)$ for all $a \in A$.  The adversary inference model~\eqref{eq:adversary_model_unichain}, can be written as $b_{t+1}=\sum_a\Tilde{\mu}(a) T^a b_t$, that implies $M_{adv}=\sum_a\Tilde{\mu}(a) T^a$.
This model captures a MDP with transition matrix $T$, with  a stationary stochastic policy $\Tilde{\mu}$ that is state independent, i.e., $\Tilde{\mu}(s,a)=\Tilde{\mu}(a)$ for all $s \in S$. Hence, since  $T$ is transition function of the unichain MDP $M$, according to the Definition \ref{def:unichain}, the  policy  $\Tilde{\mu}$, must induce an ergodic Markov chain with transition matrix $M_{adv}$. 
\end{proof}
Subsequently, for the adversary inference model~\eqref{eq:adversary_model_unichain}, 
there  exists a unique stationary distribution $b_{\infty} \in \mathbb{R}^{n}$, satisfying:
\begin{align} \label{eq:belief_station_dist}
    b_\infty =M_{adv} b_\infty, 
\end{align}
 and $\lim_{t \to \infty} b_t = b_{\infty}$, independent of $b_0$. We define asymptotic location privacy  based on the adversary  stationary belief given in below.
\begin{definition}
 Given user a mobility pattern modeled by a unichain MDP $M=(S,A,T,p_0,u)$, and a set of secrete  POIs, $S_{s} \subset S$, $M$ is asymptotically $\epsilon-$private under policy $\mu$, if
  \begin{equation}\label{eq:Asym_CSO_opacity}
\lim_{t \to \infty} \sum_{s \in S_{s}} {b_t(s)}   \le \epsilon.
 \end{equation}
\end{definition}

Due to ergodic property of $M_{adv}$, asymptotic $\epsilon-$privacy requirement~\eqref{eq:Asym_CSO_opacity} can be simply written as $A_s b_{\infty} \le \epsilon$.
This requirement characterizes a region in the adversary belief space that the user's secrete  POIs are considered safe. Therefore, even if the adversary prior belief is not in this region, i.e., $A_s b_0 \nleq \epsilon$, the LPPM can still attempt to  achieve  asymptotic $\epsilon-$privacy requirement~\eqref{eq:Asym_CSO_opacity}, by designing a policy $\mu$ that drives the adversary belief $b_t$ to eventually reaches asymptotic $\epsilon-$private region, $A_s b_\infty \le \epsilon$, and stays there. Formally, policy $\mu$ can enforce user mobility model $M$ to be asymptotically $\epsilon-$private, if there exist a stationary distribution $p_{\infty}$ satisfying~\eqref{eq:stationary_dist}, and  the following set is not empty. 
\begin{align} \label{eq:asym_CSO_const}
  \Omega_{b_{\infty} }(\theta)= \{&  b_{\infty} \in \mathbb{R}^n \mid  b_{\infty}\in \Delta(S),\\
   & b_{\infty} M_{adv} = b_{\infty} ,    A_s b^T_{\infty} \le \epsilon \}. \nonumber
\end{align}

In this setup, in order to minimize the average quality loss $v_\mu$,  we  design  an optimization problem with objective function~\eqref{eq:v_mu_average_simple}, and consider~\eqref{eq:asym_CSO_const} as a constraint.
The main challenge, however, is the bilinear form of~\eqref{eq:belief_station_dist} over  $b_\infty$ and $\theta$, that makes the problem non-convex. Formally:
\begin{align} \label{eq:BMI_design}
  &  \min \limits_{\theta, M_{adv}, b_{\infty} } v_\mu\\
  & \text{s.t. } \eqref{eq:asym_CSO_const}, \eqref{eq:M_p_infty,mu_unichain}, \eqref{eq:state-action_stionary_dist_const}, \sum_{s,a}\theta(s,a)= 1, \theta \ge \mathbf{0}. \nonumber
\end{align} 
Similarly, if $\theta^*$ is the optimal solution of \eqref{eq:BMI_design}, then for any $s\in S$, and $a \in A$, it  induces  $p^*_{\infty}(s)= \sum_a \theta^*(s,a)$, and the optimal stationary policy $\mu^*(s,a)= \theta^*(s,a)\slash p^*_{\infty}(s)$. 

\subsection{Computation Overhead}
The proposed LPPM framework designs the privacy-preserving policies  based on the linear programming~\eqref{eq:LP_policy_design}, where  the dimension of the problem is defined based on the
number of optimization variables, $d$, and the  
number of constraints $l$. In~\eqref{eq:LP_policy_design}, 
the constraints~\eqref{eq:CSO_NSCond_them} and~\eqref{eq:M_p_infty,mu_unichain} can be merged into a single inequality. Therefore,  we have  $l=|S|\times|A|+2|S|+1$ constraints, and $d=|S|\times|A|+|S|+1$ optimization variables.
The most common tool in practice to solve linear programming is simplex algorithms that is quite efficient.
In particular, the number of iterations seemed polynomial in $d$ and $l$ \cite{spielman2004smoothed}, but the worst-case complexity is proven to be exponential \cite{klee1970good}.

The asymptotic privacy-preserving synthesis problem~\eqref{eq:BMI_design}  is a bilinear matrix inequality (BMI) optimization problem.  It is proven that a general BMI constrained optimization problem is NP-hard, but despite this theoretical barrier, various approaches have been developed in the literature to tackle these optimization problems. BMI optimization problems can be solved by forming a sequence of semidefinite programming (SDP) relaxations  \cite{vandenberghe1996semidefinite}, or with other general nonlinear optimization methods such as sequential quadratic programming (SQP) \cite{boggs1995sequential}.

\section{Experimental Results} \label{sec:case_study}
In this section, we study  effectiveness of the proposed LPPM on a realistic case study involving publicly available data. Particularly, we select Geolife database \cite{zheng2010geolife} to conduct our experiment.
\subsection{Dataset Description}
This database is a GPS trajectory dataset collected by Microsoft Research Asia over three years by 182 users.
During this project, a wide range of users' outdoor movements are recorded, including daily life routines like commuting to work, and activities such as shopping, dining, and cycling.
The users' movements in this dataset are represented by a series of tuples containing latitude, longitude, and timestamps. The dataset contains mobility traces with the total number of  $25$ million locations, however, the variance of the number of locations per user and the total duration of the trajectories are very high. In our experiment, we have filtered the dataset to keep only users' trajectories with more than 500 locations and a duration of at least 1 year. After this process, the final dataset contains 18 users' data, and among them, we picked the user $\#114$ who has a large set of mobility traces.

\subsection{MDP Construction Process}
The MDP construction algorithm is depicted in Algorithm \ref{alg:MDP_construction}.
We first extract the user's POIs from his mobility traces by adapting the clustering algorithm called Density-Joinable cluster (DJ-Cluster) \cite{gambs2012next} for our experiment.
Following this method, %\cite{primault2014differentially,gambs2012next}, 
a user's POI is defined by the centroid of an area where the user frequently visits and spends a given amount of time. The maximum diameter of this  area, denoted by $\text{MaxRadi}$, the minimum duration of stay in this area, denoted by $\text{MinStay}$, and the maximum distance between the areas, denoted by $\text{MinDist}$, are required parameters to characterize the user's POIs.

DJ-Cluster algorithm has three phases. The first phase is to preprocess the user's mobility traces to extract the stationary points. 
Given a predefined constant $\text{MinSpeed} \ge 0$, we delete all the traces in which the user's speed is greater than  $\text{MinSpeed}$. The second phase is to construct a set of clusters from the remaining mobility traces and then merge the user's locations which are in $\text{MaxRadi}$ radius of the clusters' centroid.
The third phase is to merge the computed clusters in which their centroid are within $\text{MinDist}$ distance. Once, the clustering process is finalized, for each cluster $c$, we compute the user spending time in each cluster in hours, and  denote it by $\text{stay}(c)$. Given a predefined positive constant $\text{MinStay} > 0$, we remove any cluster $c$ in which  $\text{stay}(c) <\text{MinStay}$. The remaining clusters form the user's POIs that represented by the states set $S$ in the MDP $M$. The described procedure is shown in the first part of the Algorithm \ref{alg:MDP_construction} (lines 1-4). 

Figure \ref{fig:casestudy_POIs} illustrates the user mobility traces. We set $MinSpeed=0$ to extract user stationary points, and use orthodromic distance to compute the distance between the stationary points to form a cluster. The extracted  POIs  are shown in Figure \ref{fig:casestudy_POIs} with red marks representing the centroid of each cluster. Without loss of generality, here we assume the initial distribution $p_0(s_1)=1$, otherwise $p_0(s)=0$.  
 
The next is to define the MDP actions set $A$. As it is discussed in Section \ref{subsec:Loc_Release_Mech}, we consider the location releasing mechanisms as a set of actions. Here, we use spatial cloaking techniques to conceal the user's POIs. In particular, we utilize the $k-$area cloaking method proposed in \cite{gruteser2004protecting}. This technique provide location privacy by blending the user secrete location into a region that covers at
least $k-1$ other sensitive areas. Thus, the user's secret location at each  POI area becomes indistinguishable among other $k$ POIs. Here, we consider all the extracted POIs are potentially sensitive areas. For a sake of simplicity, we set  $k=2$, and consider a curricular cloaking area. Following $k-$area cloaking method \cite{gruteser2004protecting}, for any POI $s \in S$, we construct a  cloaked region $a$, with radius $r_a$, such that it contains $s$, and  at least another POI $s' \in S$. The $k-$area cloaking method iterates over each POI to find the smallest $r_a$ that meets the $2-$area cloaking privacy requirement.
Figure \ref{fig:casestudy_POIs} illustrates the constructed cloaked areas for two POIs that are represented by blue circles.
% Note that our LPPM framework can accept other location releasing mechanism as long as it can be modeled as a set of discrete decisions or events. For instance, if the user  sends dummy locations to the LBS to conceal his POIs \cite{lu2008pad}, set $A(s)$ would represent the dummy locations available at each POI $s$. 
The POIs extracted from the user mobility traces are considered as the user's true location and are not accessible by the adversary. The cloaking areas however, represent the user's service query and hence are observable to the adversary. Given any pairs of $s_i,s_j \in S$, and a cloaked region $a \in A$, the transition probability $T(s_i,a,s_j)$ can be interpreted as the probability that the  user reports area  $a$  to the LBS, and transits from $s_i$ to $s_j$.  Following \cite{gambs2012next}, we first compute $p_{ij}=\text{Prob}\left[s_j\mid s_i\right]$ as the frequency of the user traveling  from $s_i$ to $s_j$, i.e. $p_{ij}=\frac{n_{ij}}{N_i}$, where $N_i$ is the total number of time the user traveled from state $s_i$, and $n_{ij}$ is number of time the user has traveled from $s_i$ to $s_j$. Then, for any $a\in A$, and $s_i,s_j \in S$, we define $T(s_i,a,s_j)=p_{ij}$ if the cloaking area $a$ conceals $s_i$, and otherwise $T(s_i,a,s_j)=0$.

The quality loss function for this model is associated with the area of the cloaked regions. More precisely, 
we assume the user is querying for information only when he is inside the POIs area, and the $k-$area cloaking mechanism obfuscates the user's location in the query by increasing the area of retrieval. Hence the user's quality of service is expected to degrade proportionally with the area of the reported cloaked region \cite{ku2009privacy}. Therefore,  for any POI $s$, and a cloaked region $a \in A$, the utility function can be defined by $u(s,a)=\frac{\text{Area}(a)}{\text{Area}(s)}$. Note that, 
if $a$ is not available at  $s$, the LPPM policy must avoid reporting  $a$ when the user is inside $s$, and hence, we set $u(s,a)= \overline{u}$, where $\overline{u} >> \max_{s,a}(u(s,a))$ representing a large quality loss.  The described procedure is depicted in the second part of the Algorithm \ref{alg:MDP_construction} (lines 5-9).

\DontPrintSemicolon
\begin{algorithm}[!t]  
 \caption{MDP construction algorithm.}
  \label{alg:MDP_construction}
\SetAlgoLined 
\KwIn{$\text{MinSpeed} > 0,\text{MinDist} > 0,\text{MinStay} > 0,\text{MaxRadi}>0,\text{MinDist}>0$, ${k}$.}
\KwResult{User mobility MDP $ M=(S,A,T,p_0,u)$}
%\begin{enumerate}[leftmargin=*]
 Delete all the traces that user's speed $> \text{MinSpeed}$, and store  the stationary points at $\text{StPoint}$ set.\;
 Merge all points in $\text{StPoint}$ which are in $\text{MaxRadi}$ distance of each other and store as a cluster set $\mathcal{C}$.\;
 Merge any pair of clusters  $c,c' \in \mathcal{C}$ which  their distance is $\text{dist}(c,c') < \text{MinDist}$. \;
 Compute user spending time, $\text{stay}(c)$, for each cluster $c \in \mathcal{C}$, and remove all the clusters which  $\text{stay}(c) < \text{MinStay}$, and store $ \mathcal{C}$ as the  state set $S$.\;
 Set the initial distribution $p_0(s)$.\;
 For all $s_i \in S$, construct a cloaked region $a_i$ with radius $r_{a_i}$ such that the $k$-area clocking criteria is met \cite{gruteser2004protecting}. Store all $a_i$ in set $A$.\;
 For any pair of states $s_i,s_j \in S$, compute the transition probability $p_{ij}$.\;
 For any cloaked region $a_l \in A$, and any pair of states $s_i,s_j \in S$, compute the MDP transition probability by:
\begin{flalign*}
  T(s_i,a_l,s_j)= \begin{cases} p_{ij}& \text{if } s_i \in a_l \\ 0 & \text{otherwise}, \end{cases}    &&
\end{flalign*} 
where $s_i \in a_l$ means the all the points of cluster $s_i$ are inside the cloaked region $a_l$.\; 
 For any state $s \in S$, and cloaked region $a \in A$, compute the quality loss function as $u(s,a)=\frac{\text{Area}(a)}{\text{Area}(s)}$, if $s \in a$, and otherwise  set $u(s,a)=\overline{u}$.\;
%\end{enumerate}
\end{algorithm}

 \begin{figure*}
 \begin{minipage} {.32\textwidth}
 \includegraphics[width=\linewidth]{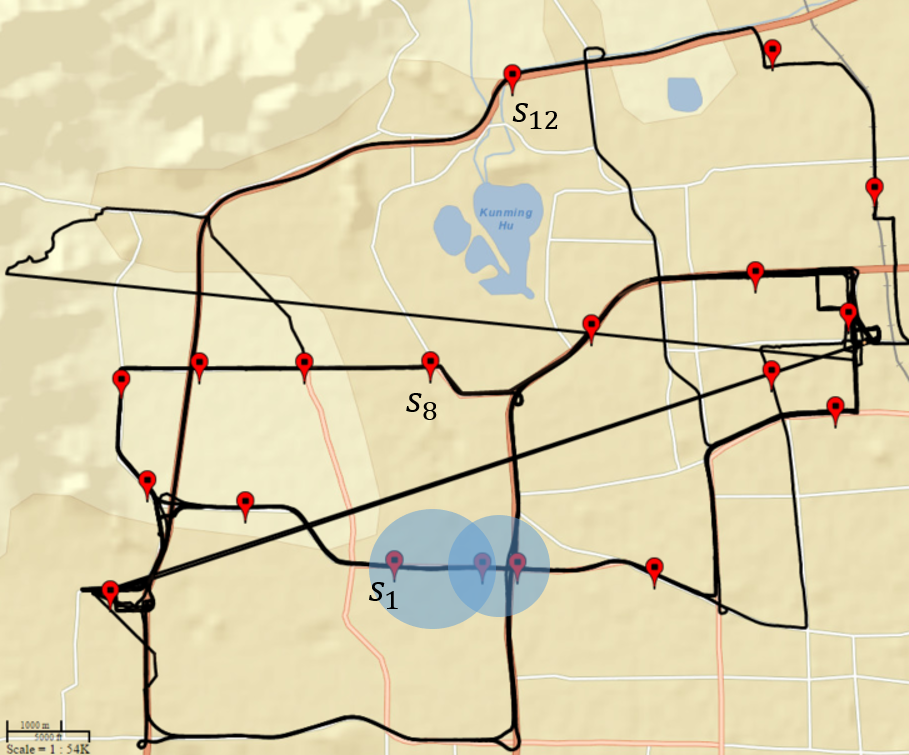}
  \caption{Example of real-world mobility traces of a mobile user. The solid black line represents the traces of user movements, the red marks are the extracted POIs, and the blue circles are examples of the cloaked areas. %For clarity of the picture, only  two cloaked areas are demonstrated.
  }
  \label{fig:casestudy_POIs}
 \end{minipage} \hfill
  \begin{minipage}{.32\textwidth}
    \includegraphics[width=\linewidth]{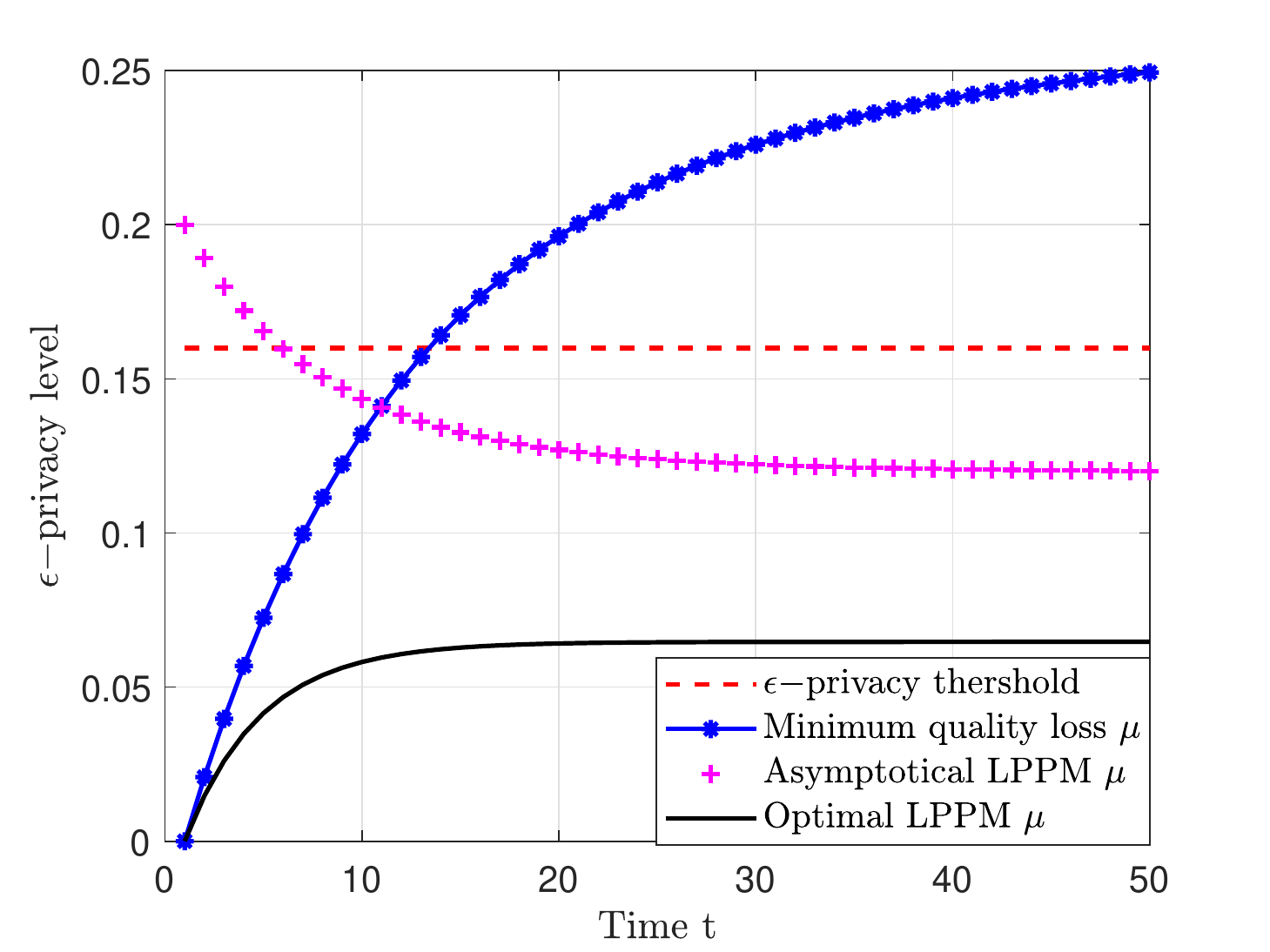}
  \caption{ 
  %The user's privacy level on the secrete location for the different policy.
 The dashed line represents $\epsilon$. Respectively, the blue dash-dot, black solid, the purple plus-sign  lines represent  $b_t(S_s)$  for the optimal policy $\mu$,  that provides the best utility  without considering $\epsilon-$privacy, and $\mu$ that guarantees $\epsilon-$privacy,  and asymptotic  $\epsilon-$privacy requirements. 
}
  \label{fig:casestudy_CSO_level}
  \end{minipage} \hfill
  \begin{minipage}{.32\textwidth}
   \includegraphics[width=\linewidth]{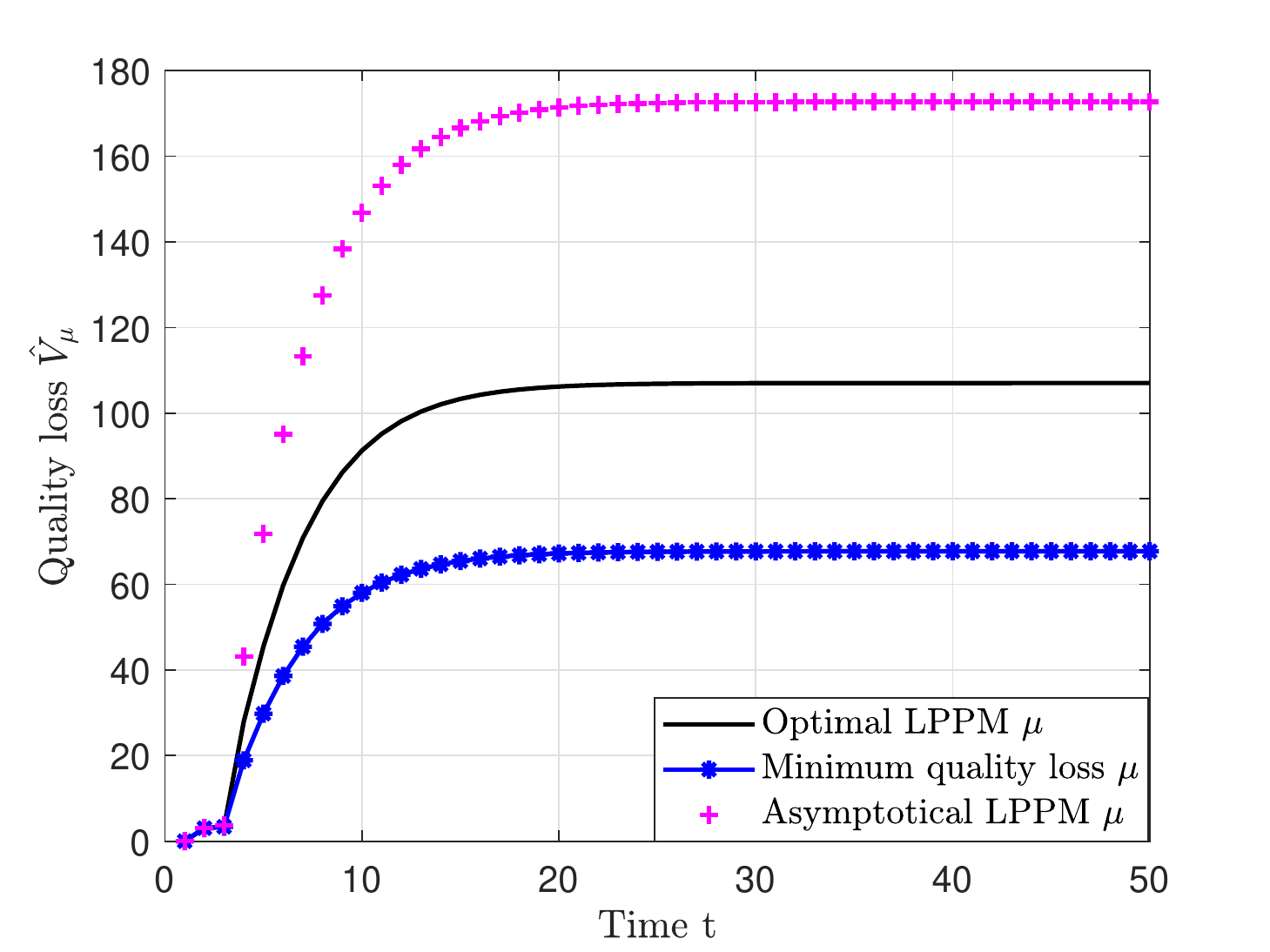}
  \caption{Respectively, the blue dash-dot, black solid, the purple plus-sign  lines represent $\hat{v}_\mu(t)$ for the optimal policy $\mu$, that provides the minimum quality loss without considering  $\epsilon-$privacy  requirement,  and $\mu$ that guarantees $\epsilon-$privacy,  and asymptotic  $\epsilon-$privacy requirements.}
  \label{fig:casestudy_qualityloss_level}
  \end{minipage} 
 
 \end{figure*} 
 
 \subsection{Simulation with Real-world Dataset}
 The constructed  MDP has $19$ states, and $19$ actions, and all  the states respect $k-$area cloaking privacy requirement \cite{gruteser2004protecting}. We found one of the user's extracted POI is nearby a bank, labeled by $s_{12}$, and another one is nearby a bus station, labeled by $s_8$. We consider them as the user's secret POIs, and therefore, we define $S_s=\{s_8,s_{12}\}$.  
 We assume the adversary knows the user's mobility model $M$, and maintains a prior belief over the user's POIs. 
 The adversary observes the reported cloaked areas and uses inference model~\eqref{eq:adversary_general_model} to infer the user presence at $S_s$. 
 Quantify by $\epsilon-$privacy metric~\eqref{eq:CSO_opacity}, the LPPM protects user privacy by randomizing the selection of cloaked regions to suppress the adversary belief on $S_s$.

 Our objective here is  to evaluate the privacy level and the expected total quality loss when the user incorporates the policy $\mu$ in his location releasing mechanism.   Let's consider $\epsilon = 0.16$ as the desired privacy level, and denote $b_t(S_s)=\sum_{s \in S_s}b_t(s)$.  In order to demonstrate  the adversary can improve his knowledge over the user's secrete locations, we set $b_0(s)=0$ for all $s \in S_s$, and
 $b_0(s)=1/|S -  S_s|$ for any $s \in \{S -  S_s\}$,  which also indicates  $A_s b_0 < \epsilon$. 

 We first design a policy that minimizes the user's average quality loss $v_\mu$ without considering the $\epsilon-$privacy requirement. In particular, we use our proposed LP~\eqref{eq:LP_policy_design} without  $\epsilon-$privacy constraint~\eqref{eq:CSO_opacity}. 
 As it can be seen in  Figure \ref{fig:casestudy_CSO_level}, although initially $M$ is $\epsilon-$private, the adversary belief  over secret state $S_s$, increases over time and eventually violates the desired $\epsilon-$privacy condition. Intuitively, the adversary becomes more confident over the user presence at the area of $S_s$, although all the user's POIs are concealed individually with the cloaking regions. The average quality loss associated with this policy is $v_\mu=67.72$.
 We then synthesize the privacy-preserving policy using LP~\eqref{eq:LP_policy_design} with the privacy level $\epsilon=0.16$. Demonstrated by Figure \ref{fig:casestudy_CSO_level}, the LPPM can suppress the adversary belief to meet the desired privacy level. This level of user privacy, however, comes with a price. The average quality loss has increased to $v_\mu=107.03$ when the user uses the privacy-preserving policy. To further illustrate this trade-off, let's define the user's quality loss at time $t$ by $\hat{v}_\mu(t)=\sum_{s,a}p_t(s)\mu(s,a)u(s,a)$. Figures \ref{fig:casestudy_qualityloss_level}  shows how $\hat{v}_\mu(t)$  evolves over time for different policies.   As it is depicted in Figures \ref{fig:casestudy_CSO_level} and \ref{fig:casestudy_qualityloss_level},  although the proposed LPPM can suppress the adversary belief over the secret locations, the user continuously receives lower quality of service when the privacy-privacy policy is incorporated in the LPPM. 
 
  Now let's consider a scenario that the adversary prior belief violates the $\epsilon-$privacy condition,   implying the adversary background knowledge over the user presence at $S_s$ is unsafe.  In this case,  the LPPM can still deceive the adversary by manipulating her belief over secrete locations to eventually satisfy the defined $\epsilon-$privacy criteria. 
Let's assume the adversary prior belief over the user's secret POIs is $\sum_{s \in S_s} b_0(s)=0.2$, and for other POIs $s' \in \{S - S_s\}$, is defined as $b_0(s')=0.8 \slash |S - S_s|$, indicating the violation of defined $\epsilon-$privacy, i.e., $A_s b_0 \not \le \epsilon$.  The LPPM objective here is to design the policy $\mu$ that forces the adversary  belief $b_t$ to meet the  asymptotic $\epsilon-$privacy requirement~\eqref{eq:Asym_CSO_opacity}.
To achieve this, we synthesize  $\mu$ using the proposed BMI given by~\eqref{eq:BMI_design}.
The trajectory of $b_t(S_s)$, and the associated $\hat{v}_\mu(t)$ are respectively shown in Figure \ref{fig:casestudy_CSO_level} and \ref{fig:casestudy_qualityloss_level}. As it is demonstrated, although the adversary prior belief violates the $\epsilon-$privacy condition, the LPPM can deceive the adversary by gradually suppressing his belief over the user's secrete  POIs  $b_t(S_s)$, and asymptotically satisfies the desired $\epsilon-$privacy requirement.

The proposed privacy-preserving policy synthesis  is simulated in MATLAB on a PC with
Intel(R) Core(TM) i7-8650 CPU 1.9GHz 16GB RAM running on Windows 10 professional OS. 
Finding the optimal $\epsilon-$privacy LPPM,  expressed as a LP~\eqref{eq:LP_policy_design}, took $0.46$ second, and synthesis of the optimal asymptotic $\epsilon-$privacy LPPM, given as a BMI problem~\eqref{eq:BMI_design}, took $29.53$ seconds.

\section{Conclusion}\label{sec:conclusion}
In this paper, we have designed and demonstrated a model-based privacy-preserving framework that guarantees a user-defined privacy requirement for an infinite time horizon while minimizing the quality loss of service received by the user.  In this regard, a MDP is constructed to capture the user mobility pattern and the LBS utility model. Given the MDP with state representing the user's locations, we adapt the probabilistic current-state opacity notion to introduce new location privacy notion, $\epsilon-$privacy, which characterizes the user privacy against a Bayesian adversary with localization attack model. Through this setup, we illustrated that even if each user location is concealed from the adversary, she still can utilize the user mobility model to further reduce his uncertainty over the user's secret locations. Given this privacy concern, we developed a LPPM that randomizes the obfuscation mechanisms to protect user privacy against the adversary with such an inference capability. The overall privacy-preserving framework is demonstrated and validated on an experimental dataset.

\ifCLASSOPTIONcompsoc
  \section*{Acknowledgments}
\else
  \section*{Acknowledgment}
\fi

This work was supported in part by the National Science Foundation under Grant Grant IIS-1724070, and Grant CNS-1830335, and in part by the Army Research Laboratory under Grant W911NF-17-1-0072.

\ifCLASSOPTIONcaptionsoff
  \newpage
\fi



%\bibliographystyle{IEEEtran}
\bibliographystyle{IEEEtran}
\bibliography{LBS_Privacy_Preserving_Ref.bib}

\vspace{-10 mm}
\begin{IEEEbiography}[{\includegraphics[width=1in,height=1.25in,clip,keepaspectratio]{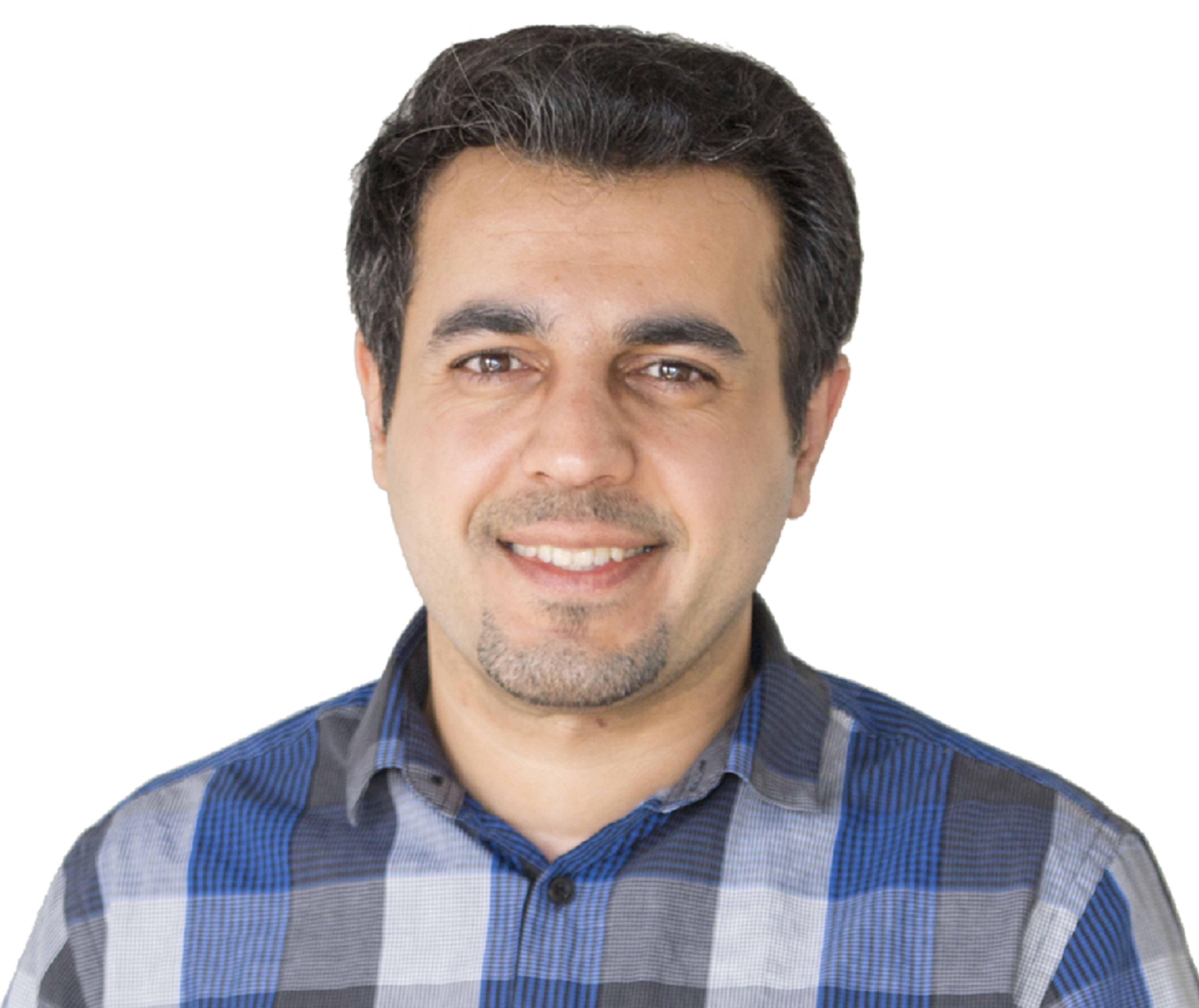}}]{Alireza Partovi}
received his M.E. in Electrical Engineering from National University of Singapore  and in Control System Engineering from University of Notre Dame. He is currently working toward Ph.D. degree at the University of Notre Dame. His current research interests include formal
methods and control of probabilistic systems with
application in cyber-security and multi-robot
systems.
\end{IEEEbiography}
\vspace{-10 mm}
\begin{IEEEbiography}[{\includegraphics[width=1in,height=1.25in,clip,keepaspectratio]{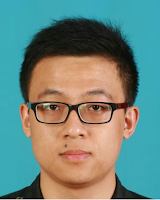}}]{Wei Zheng}
received his B.E. degree in Automation and M.E. degree in Control Science and Engineering from the University of Science and Technology of China, Hefei, China, in 2012 and 2016, respectively. He is currently working toward Ph.D. degree at the University of Notre Dame.
\end{IEEEbiography}
\vspace{-10 mm}
%\newpage
\begin{IEEEbiography}[{\includegraphics[width=1in,height=1.25in,clip,keepaspectratio]{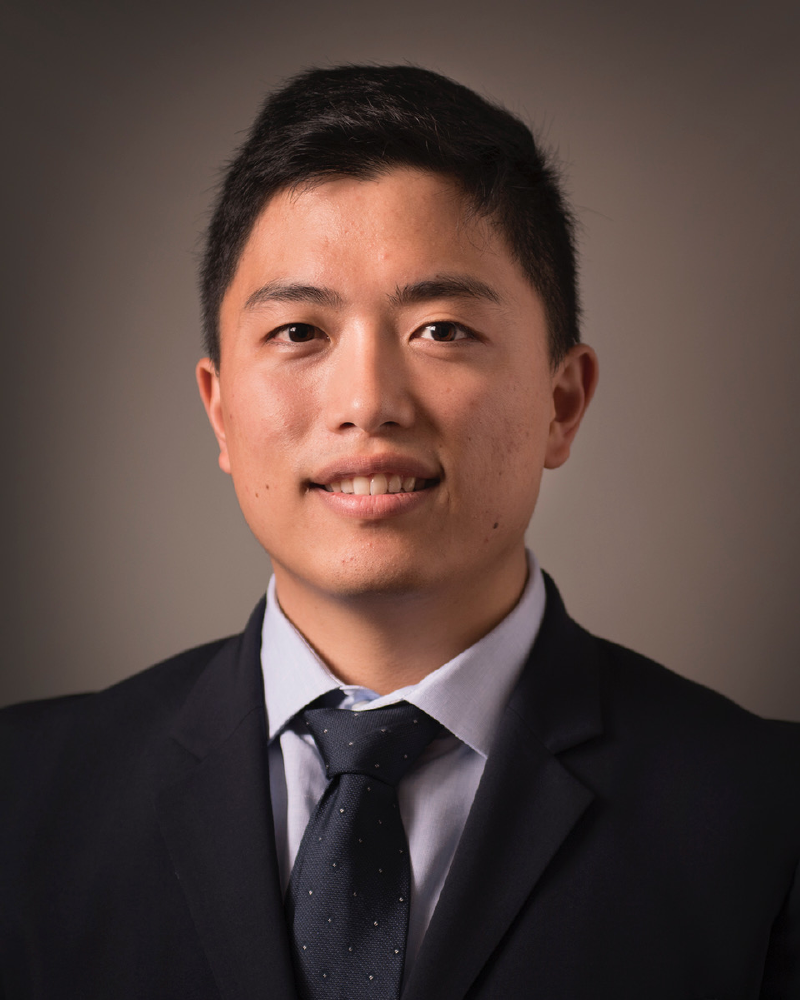}}]{Taeho Jung}
 is an assistant professor of Computer Science and Engineering at the University of Notre Dame. He received the Ph.D. from Illinois Institute of Technology in 2017 and B.E. from Tsinghua University in 2011. His research area includes data security, user privacy, and applied cryptography. His paper has won a best paper award (IEEE IPCCC 2014), and two of his papers were selected as best paper candidate (ACM MobiHoc 2014) and best paper award runner up (BigCom 2015).
\end{IEEEbiography}
\vspace{-10 mm}
\begin{IEEEbiography}[{\includegraphics[width=1in,height=1.25in,clip,keepaspectratio]{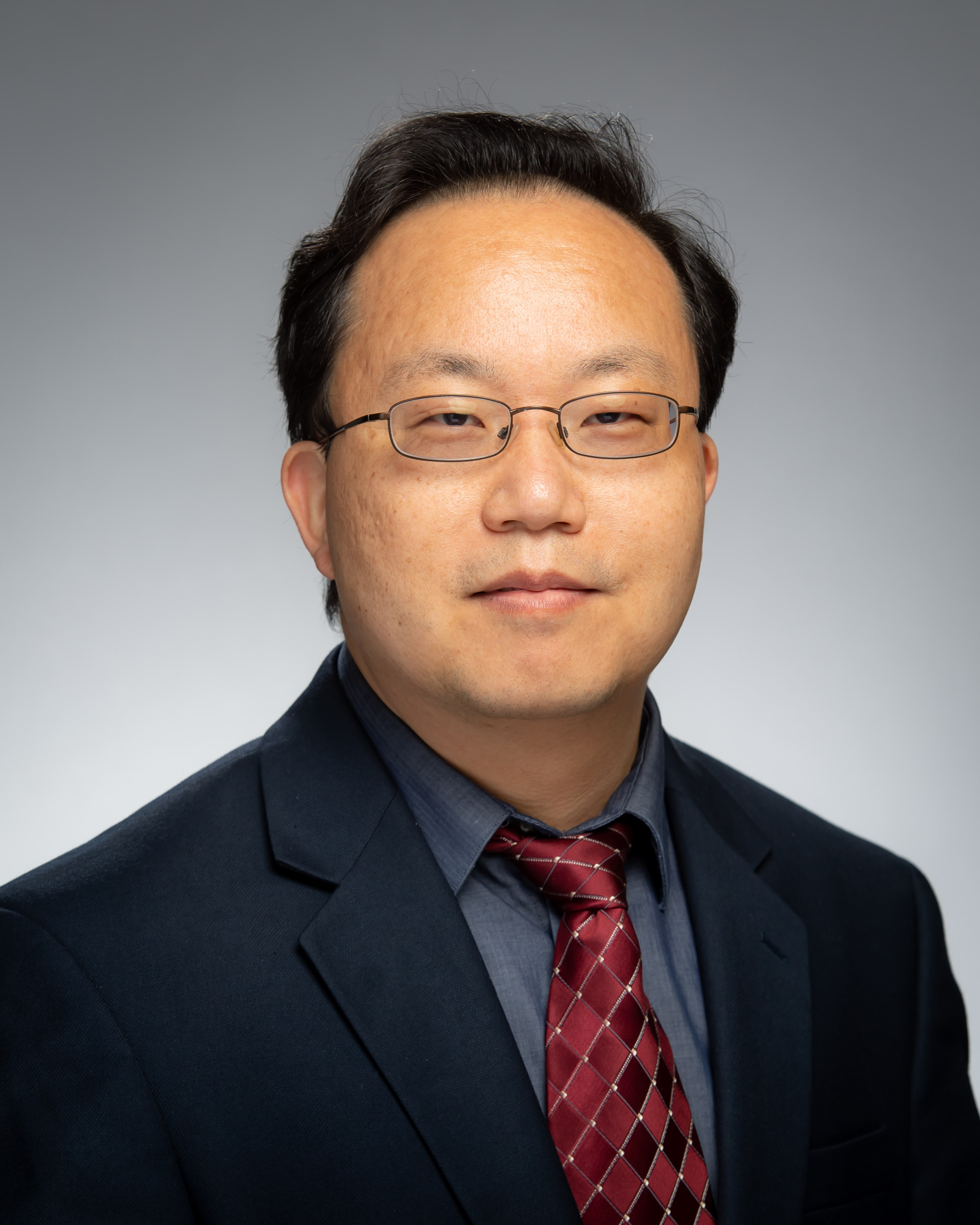}}]{Lin Hai}
 is currently a professor at the Department of Electrical Engineering, University of Notre Dame, where he got his Ph.D. in 2005. Before returning to his alma mater, Hai has been working as an assistant professor in the National University of Singapore from 2006 to 2011. Dr. Lin's teaching and research interests are in the multidisciplinary study of the problems at the intersections of control, machine learning and formal methods. His current research thrust is on cyber-physical systems, multi-robot cooperative tasking, human-machine collaboration, and security/privacy. Hai has been served in several committees and editorial board, including IEEE Transactions on Automatic Control. He served as the chair for the IEEE CSS Technical Committee on Discrete Event Systems, program chair for IEEE ICCA 2011, IEEE CIS 2011 and the chair for IEEE Systems, Man and Cybernetics Singapore Chapter for 2009 and 2010. He is a senior member of IEEE and a recipient of 2013 NSF CAREER award.
\end{IEEEbiography}
\vfill


%\vfill


\end{document}